\documentclass{article}

\usepackage[leqno,intlimits]{amsmath}
\usepackage{amssymb}
\usepackage{amsthm}
\usepackage{graphicx}
\usepackage{caption}
\usepackage{subcaption}
\usepackage{hyperref}
\usepackage{url}
\usepackage[hyphenbreaks]{breakurl} % For breaking long urls in the references.

\usepackage[margin=1in]{geometry}

\allowdisplaybreaks[1]

\newcommand{\pd}[2]{\frac{\partial#1}{\partial#2}} % partial derivatives quicker.
\newcommand{\E}{\mathbb{E}} % Expectation E
\newcommand{\p}{\mathbb{P}} % Probability P
\newcommand{\eps}{\varepsilon} % Epsilon
\DeclareMathOperator*{\st}{st}
\DeclareMathOperator*{\cx}{cx}
\DeclareMathOperator*{\icx}{icx}
\DeclareMathOperator*{\Bin}{Bin}

\DeclareMathOperator{\R}{\mathbb{R}}

\DeclareMathOperator{\mcF}{\mathcal{F}}

\def\bfmath#1{\mathchoice
        {\mbox{\boldmath$#1$}}%
        {\mbox{\boldmath$#1$}}%
        {\mbox{\boldmath$\scriptstyle#1$}}%
        {\mbox{\boldmath$\scriptscriptstyle#1$}}}%

\def\bfa{\bfmath{a}}

\def\bfu{\bfmath{u}}

\def\bfx{\bfmath{x}}
\def\bfy{\bfmath{y}}

\def\bfA{\bfmath{A}}

\def\bfF{\bfmath{F}}

\def\bfM{\bfmath{M}}

\def\bfP{\bfmath{P}}

\def\bfR{\bfmath{R}}

\def\bfX{\bfmath{X}}

\def\bfZ{\bfmath{Z}}

\def\cA{{\mathcal A}}
\def\cB{{\mathcal B}}

\def\cF{{\mathcal F}}

\newtheorem{theorem}{Theorem}[section]
\newtheorem{lemma}[theorem]{Lemma}

\newtheorem{corollary}[theorem]{Corollary}
\newtheorem{conjecture}[theorem]{Conjecture}
\newtheorem{definition}{Definition}  % Note, this italicizes everything

\begin{document} %%%
%%%%%%%%%%%%%%%%%%%%

%%%%%%%%%%%%%
%%% Title %%%
%%%%%%%%%%%%%

\title{Coexistence in preferential attachment networks}
\author{
	Ton\'ci Antunovi\'c
	\thanks{University of California, Los Angeles; \texttt{tantunovic@math.ucla.edu}.}
	\and
	Elchanan Mossel
	\thanks{University of Pennsylvania and University of California, Berkeley; \texttt{mossel@wharton.upenn.edu}; supported by NSF grant DMS 1106999 and by DOD ONR grant N000141110140.}
	\and
	Mikl\'os Z. R\'acz
	\thanks{University of California, Berkeley; \texttt{racz@stat.berkeley.edu}; supported by a UC Berkeley Graduate Fellowship, by NSF grant DMS 1106999 and by DOD ONR grant N000141110140.}}
\date{\today}

\maketitle

%%%%%%%%%%%%%%%%
%%% Abstract %%%
%%%%%%%%%%%%%%%%

\begin{abstract}
We introduce a new model of competition on growing networks. This extends the preferential attachment model, 
with the key property that node choices evolve simultaneously with the network. 
When a new node joins the network, it chooses neighbours by preferential attachment, and selects its type based on the number of initial neighbours of each type. 
The model is analysed in detail, and 
in particular, we determine the possible proportions of the various types in the limit of large networks. 
An important qualitative feature we find is that, in contrast to many current theoretical models, often several competitors will coexist. 
This matches empirical observations in many real-world networks. 
\end{abstract} 


%\maketitle

%%%%%%%%%%%%%%%%
%%% Document %%%
%%%%%%%%%%%%%%%%

%%%%%%%%%%%%%%%%%%%%%%%%%%
\section{Introduction} %%%
%%%%%%%%%%%%%%%%%%%%%%%%%%

A major challenge in  complex networks is understanding the interplay between the evolution of the network and the dynamical processes that take place on it. 
Many networks evolve dynamically, e.g., 
the citation graph grows every day with new papers being published, 
and friendships are created and broken every minute. 
The changes in network structure are closely related to 
the processes on these nodes, e.g.,  
the content of a Facebook page is correlated with the friendship dynamics. 

In the past fifteen years there have been many studies on processes on networks~\cite{barrat2008dynamical}, e.g., epidemic spreading~\cite{pastor2001epidemic}, evolutionary games~\cite{ohtsuki2006simple}, and information cascades~\cite{watts2002simple}. 
However, most considered the network as fixed, and then studied the process of interest on static graphs. 
This static viewpoint hides the fact that the networks and the processes on them \emph{coevolve}. 
Although the study of such coevolution was initiated over a decade ago~\cite{skyrms2000dynamic}, it is only recently beginning to be explored in greater depth (see~\cite{gross2008adaptive,holme2012temporal} and references therein), 
and thus many questions still remain. 
In particular, in the context of type adoption on networks, an important open problem is to understand the phenomenon of coexistence of competing types.

Here, we present a 
\emph{simple} connection which couples the growth of a network and nodal dynamics. In particular, we focus on \emph{type adoption dynamics}, where each node has a single type from a finite set of types. When a new node joins the network, both its connections to the existing nodes and its type are influenced by the current structure of the network. As a particular instance of such a general model, we consider the dynamics where the new node chooses its connections according to linear preferential attachment~\cite{barabasi1999emergence,bollobas2004diameter,BeBoChSa:14}, and then chooses its type based on how many of its neighbours are of a certain type; see Fig.~\ref{fig:model} for an illustration.

\begin{figure}[h!]%[h!]%[ht]
\centering
\includegraphics[width=0.35\textwidth]{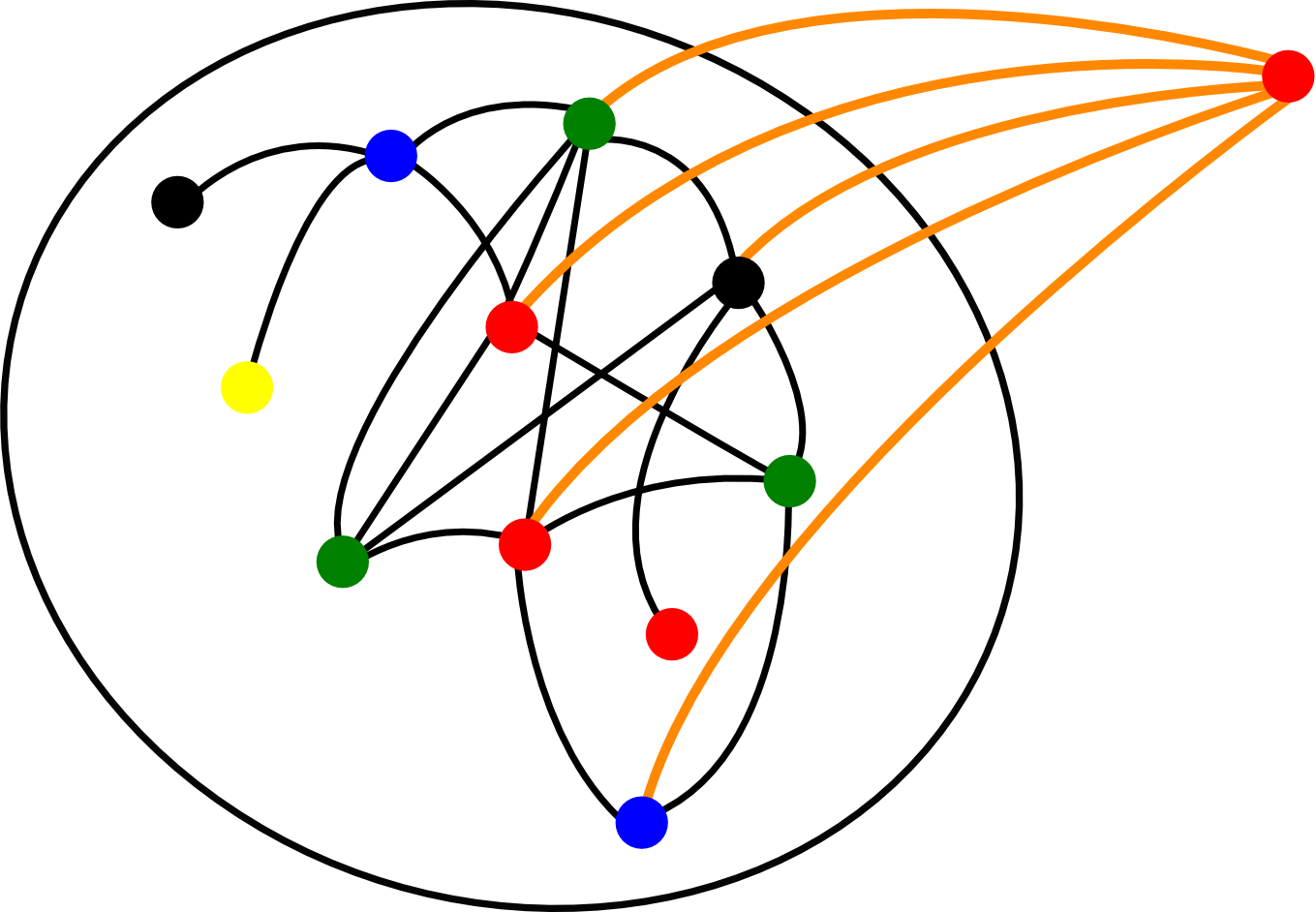}
\caption{\textbf{Illustration of the model.} Each node in the initial graph has a colour (type) from a finite set of colours. At each time step a new node is added to the graph and connected to $m$ existing nodes according to linear preferential attachment (here $m = 5$). When the new node joins the graph it also adopts a colour: it picks its colour according to a probability distribution which depends on the colours of its initial neighbours. See Section~\ref{sec:model} for details.}
\label{fig:model}
\end{figure}

This model is of interest in many cases where preferential attachment is a good representation of the evolution of the network structure and where competition between types  occurs as soon as the node joins the network. 
A natural example is the network of scientific papers linked by citations~\cite{redner1998popular}. 
Often there are two opposing viewpoints on an issue, leading to competition between the perspectives. 
Another example is that of product adoption via word-of-mouth recommendations on social networks, such as a new cell-phone user choosing a cell-phone provider based on her friends' decisions.

A key feature of this model is the simplicity of its analysis. 
We explicitly calculate the possible ratios of the types in the limit of large networks. 
The most important property of the model is that for many settings of the parameters, none of the types dominate, which matches empirical observations in many current networks. 
These results thus provide a theoretical understanding of coexistence of types in preferential attachment networks. 
They should be compared to results on other models of competition on scale-free networks where coexistence is rarely achieved, and typically the ``winner takes all''~\cite{prakash2012winner,deijfen2013winner}. 
(See also~\cite{antunovic2011competing,lelarge2012diffusion} for results on related models.) 
The chief difference between this model and related ones in economics~\cite{arthur1990positive,arthur1994increasing}, and marketing~\cite{banerjee2004word} is the explicit modeling of the underlying network structure.

We next describe the model and our results in more detail.

%%%%%%%%%%%%%%%%%%%%%%%%%%%%%%%%%%%%%%%
\subsection{Model}\label{sec:model} %%%
%%%%%%%%%%%%%%%%%%%%%%%%%%%%%%%%%%%%%%%

For simplicity, we describe the model in the case of two types, which we refer to as red and blue colours (in the following, we use the terms ``type'' and ``colour'' interchangeably). 
The model naturally generalises to any number of types; 
see Section~\ref{sec:many} for a description and results. 
The main feature of the model is that it incorporates and couples two processes: a network growing process and a type adoption process. 

We consider a natural variant of the linear preferential attachment model~\cite{barabasi1999emergence,bollobas2004diameter,BeBoChSa:14} as the network growing process.\footnote{The preferential attachment model was introduced by Barab\'asi and Albert~\cite{barabasi1999emergence}; 
however, as pointed out by Bollob\'as and Riordan~\cite{bollobas2004diameter}, 
the definition in~\cite{barabasi1999emergence} left some details of the process up to interpretation. 
There are several natural ways in which to define the process precisely: 
one option is the model introduced in~\cite{bollobas2004diameter}, 
while here we consider a slighly different variant (the ``independent model'' in~\cite{BeBoChSa:14}).} 
Starting from an initial graph $G_0$, at each time step an additional node $v$ is added to the graph, together with $m$ edges connecting $v$ to existing nodes in the graph. 
Each edge is chosen independently, and according to linear preferential attachment, i.e., the probability that a given edge connects $v$ to a given existing node $u$ is proportional to the degree of $u$.\footnote{In particular, if $m > 1$ then $v$ might connect to an existing node with multiple edges, resulting in a \emph{multigraph}. Multiple edges between two nodes can be thought of as stronger ties between the agents represented by the nodes.}

The type adoption process on the network is as follows. All nodes in the initial graph $G_0$ start with a type, i.e., they are either red or blue. Each additional node $v$ receives a colour when it is added to the graph, and this colour depends on the colours of the nodes it connects to when it is added. 
Suppose that out of the $m$ edges connecting the new node $v$ to existing ones, exactly $k$ connect to a red node. Then, conditioned on this event, $v$ becomes red with probability $p_k$ and blue with probability $1-p_k$. 
The probabilities  $p_k \in \left[0,1\right]$, where $0 \leq k \leq m$, are parameters of the model, 
and can capture a wide range of behaviour. 
A natural choice is the \emph{linear model}, when $p_k = k / m$ for all $k$. 
However, \emph{nonlinear models}, when $p_k \neq k/m$ for some $k$, can capture diminishing and increasing returns, and even more complex behaviour.

%%%%%%%%%%%%%%%%%%%%%%%%%%%%%%%%%%%%%%%%%%%
\subsection{Results}\label{sec:results} %%%
%%%%%%%%%%%%%%%%%%%%%%%%%%%%%%%%%%%%%%%%%%%

We are interested in the fraction of nodes of each type. 
This corresponds to the fraction of scientific papers sharing a viewpoint on a given topic or to the fraction of individuals using a given company's product, i.e., the company's market share. 
Our main results characterise the possible limiting fractions of the colours as the size of the network goes to infinity. 
The results thus provide a complete phase diagram of the asymptotic behaviour of the process; see Fig.~\ref{fig:phase} for an illustration. 
These show that if there is negative reinforcement, then the types coexist, while if there is strong positive reinforcement, then the winner takes all. Our results precisely determine where the transition happens, and show that for most parameter values the types coexist. 

To describe our results precisely, we introduce some notation. 
Let $G_n$ denote the graph when $n$ nodes have been added to the initial graph $G_0$,  
let $A_n$ and $B_n$, respectively, denote the number of red and blue nodes, respectively, in $G_n$, and let $a_n := \frac{A_n}{A_n + B_n}$ and $b_n := \frac{B_n}{A_n+B_n}$ denote the corresponding normalised fractions. 
Furthermore, let $X_n$ (resp., $Y_n$) denote the sum of the degrees of red (resp., blue) nodes in $G_n$, and let $x_n := \frac{X_n}{X_n + Y_n}$ and $y_n := \frac{Y_n}{X_n + Y_n}$ denote the normalised fractions. We are primarily interested in the asymptotic proportion of red and blue nodes, i.e., in the limits $\lim_{n\to \infty} a_n$ and $\lim_{n\to \infty} b_n = 1- \lim_{n \to \infty} a_n$. 

As we shall see, a key role in the asymptotic behaviour of the process is played by the polynomial 
\begin{equation}\label{eq:P}
 P \left( z \right) = \frac{1}{2} \sum_{k=0}^m \binom{m}{k} z^k \left( 1 - z \right)^{m-k} \left( p_k - \frac{k}{m} \right),
\end{equation}
and in particular its zero set, denoted by $Z_P := \left\{ z \in \left[ 0, 1 \right] : P \left( z \right)  = 0 \right\}$. 
This is because, as we shall show, $\left\{ a_n \right\}_{n \geq 0}$ behaves approximately like the solution to a stochastic version of the ordinary differential equation (ODE) $dz / dt = P \left( z \right)$, and thus intuitively the trajectory of $\left\{ a_n \right\}_{n \geq 0}$ should approximate the trajectory $\left\{ z \left( t \right) \right\}_{t \geq 0}$ of this ODE. 

The following two theorems confirm this intuition. There is an important distinction between the \emph{linear model} and \emph{nonlinear models}, which is due to the fact that in the linear model the polynomial $P$ is identically zero and thus $Z_P = \left[0, 1 \right]$, while in nonlinear models the zero set $Z_P$ is a finite set.

\begin{theorem}[Linear model]\label{thm:continuous_case}
Suppose that $p_k = k/m$ for all $0 \leq k \leq m$, and that $X_0, Y_0 > 0$. 
Then $a_n$ converges almost surely as $n \to \infty$.  
Furthermore, the limiting distribution of $a := \lim_{n \to \infty} a_n$ has full support on the interval $[0,1]$, has no atoms (i.e., for every $z\in\left[0,1\right]$, $\p \left( a = z \right) = 0$), and depends only on $X_0$, $Y_0$, and $m$.
\end{theorem}

See Figure~\ref{fig:hist} for empirical histograms in the linear model with various initial parameters and various values of $m$, 
which show that a wide variety of limiting behaviours is possible.

\begin{figure}[h!]%[ht]
    \centering
    \begin{subfigure}[h]{0.32\textwidth}
	\centering
	\includegraphics[width=\textwidth]{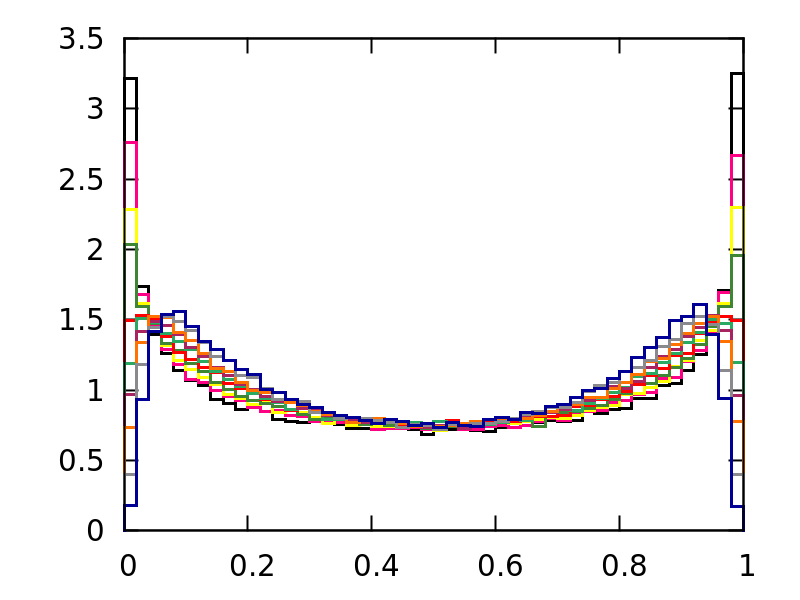}
	\caption{$A_0 = B_0 = 1, X_0 = Y_0 = 1.$}
	\label{fig:hist1}
    \end{subfigure}
    \ \ 
    \begin{subfigure}[h]{0.32\textwidth}
	\centering
	\includegraphics[width=\textwidth]{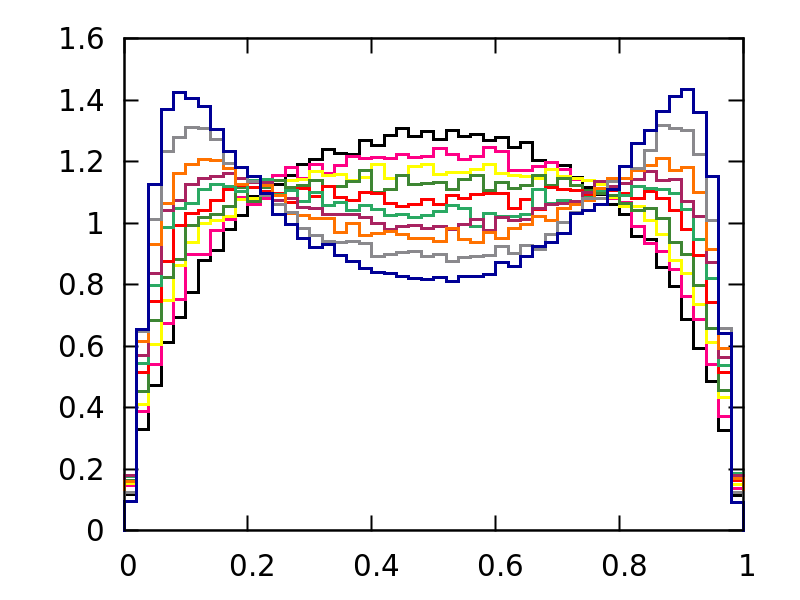}
	\caption{$A_0 = B_0 = 2, X_0 = Y_0 = 4.$}
	\label{fig:hist2}
    \end{subfigure}
    \ \ 
    \begin{subfigure}[h]{0.32\textwidth}
	\centering
	\includegraphics[width=\textwidth]{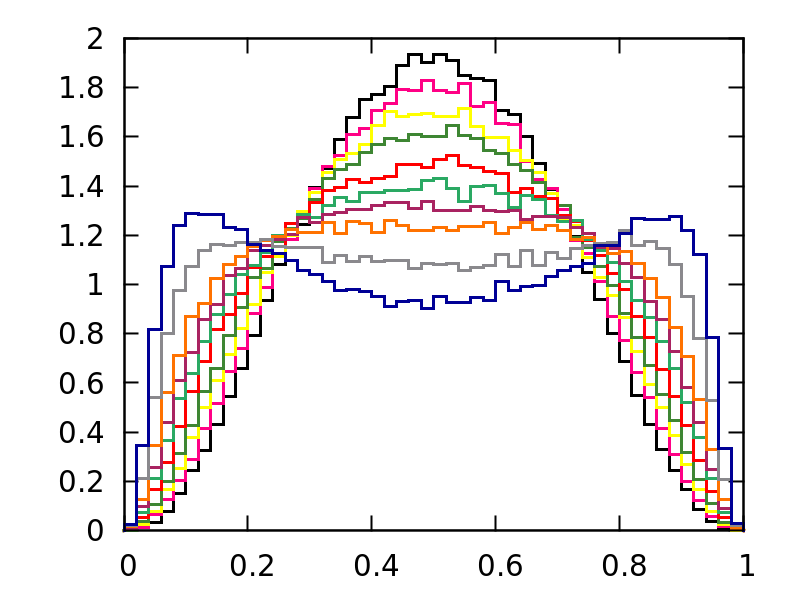}
	\caption{$A_0 = B_0 = 3, X_0 = Y_0 = 9.$}
	\label{fig:hist3}
    \end{subfigure}
    \\
    \begin{subfigure}[h]{0.32\textwidth}
	\centering
	\includegraphics[width=\textwidth]{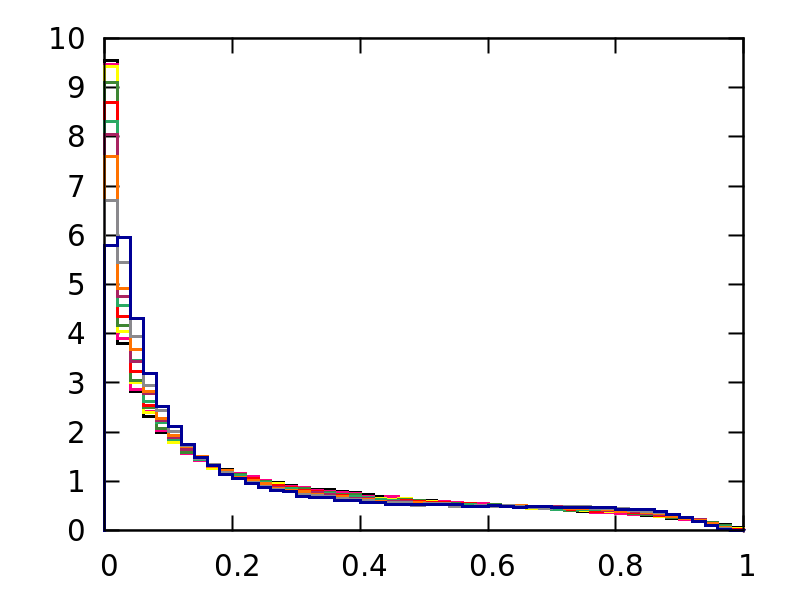}
	\caption{$A_0 = 1, B_0 = 2, X_0 = 1, Y_0 = 3.$}
	\label{fig:hist4}
    \end{subfigure}
    \ \ 
    \begin{subfigure}[h]{0.32\textwidth}
	\centering
	\includegraphics[width=\textwidth]{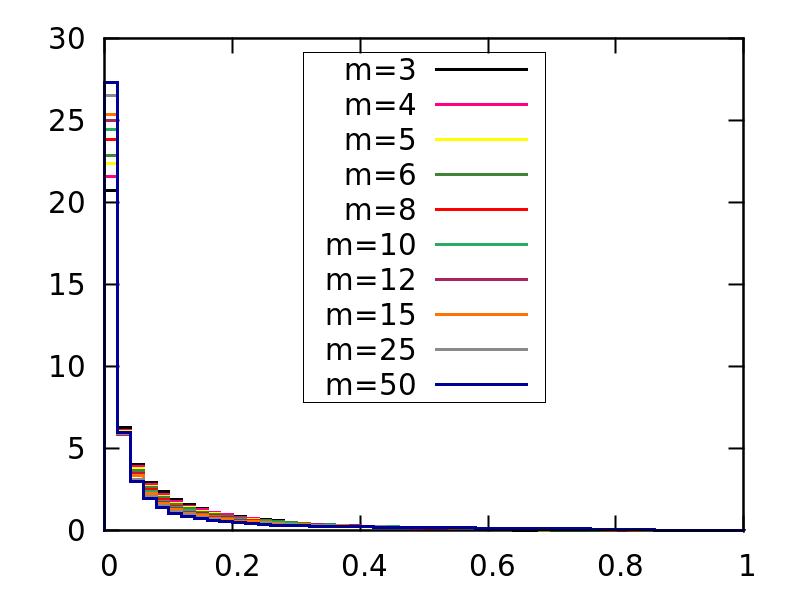}
	\caption{$A_0 = 1, B_0 = 4, X_0 = 1, Y_0 = 11.$}
	\label{fig:hist5}
    \end{subfigure}
    \ \ 
    \begin{subfigure}[h]{0.32\textwidth}
	\centering
	\includegraphics[width=\textwidth]{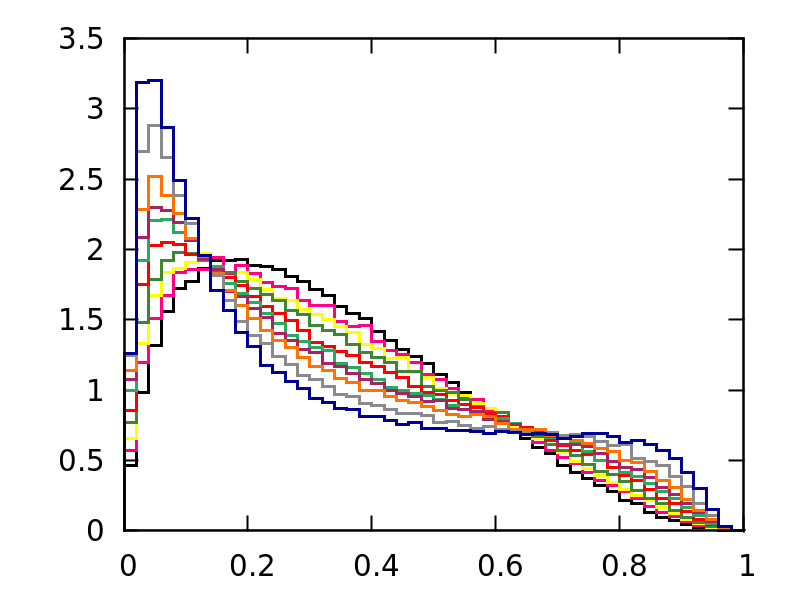}
	\caption{$A_0 = 2, B_0 = 3, X_0 = 4, Y_0 = 8.$}
	\label{fig:hist6}
    \end{subfigure}
    \caption{\textbf{Empirical histograms of $a_n$ in the linear model} for $n = 10^5$, from $2 \times 10^5$ simulations. Each subfigure has different initial parameters (see subcaptions), and in each case empirical histograms for ten different values of $m$ are plotted. See Fig.~\ref{fig:hist5} for the key to all plots.}
    \label{fig:hist}
\end{figure}

\begin{theorem}[Nonlinear models]\label{thm:point_mass_case}
Suppose that $p_k \neq k/m$ for some $0 \leq k \leq m$, and that $X_0, Y_0 > 0$. 
Then $a_n$ converges almost surely as $n \to \infty$. 
Furthermore, the limit is a point in the finite set $Z_P$.
\end{theorem}

In nonlinear models we thus know that the asymptotic proportion of red nodes is contained in the finite zero set $Z_P$. 
But which points $z \in Z_P$ arise as the limiting proportion with positive probability? 
This depends on the behaviour of the polynomial around the zero $z \in Z_P$. 
Intuitively, since $\left\{ a_n \right\}_{n \geq 0}$ is a stochastic system, we expect that stable trajectories of the ODE $dz / dt = P \left( z \right)$ should appear, but unstable trajectories should not. 
This intuition is confirmed and formalized in the following three theorems. 

\begin{theorem}[Nonlinear models, stable equilibria]\label{thm:nonlin_stable}
 Suppose that $p_k \neq k/m$ for some $0 \leq k \leq m$, and that $X_0, Y_0 > 0$. 
Suppose that $z \in Z_P \cap \left( 0, 1 \right)$ is such that there exists an $\eps > 0$ such that $P > 0$ on $\left( z - \eps, z \right)$ and $P < 0$ on $\left( z, z + \eps \right)$. 
Then $\p \left( \lim_{n \to \infty} a_n = z \right) > 0$, i.e., $a_n$ converges to $z$ with positive probability. 
Similarly, if $0 \in Z_P$ and $P < 0$ on $\left( 0, \eps \right)$, or if $1 \in Z_P$ and $P > 0$ on $\left( 1 - \eps, 1 \right)$, then there is a positive probability of convergence of $a_n$ to $0$ or $1$, respectively.
\end{theorem}

\begin{theorem}[Nonlinear models, unstable equilibria]\label{thm:nonlin_unstable}
 Suppose that $p_k \neq k/m$ for some $0 \leq k \leq m$, and that $X_0, Y_0 > 0$. 
Suppose that $z \in Z_P \cap \left( 0, 1 \right)$ is such that there exists an $\eps > 0$ such that $P < 0$ on $\left( z - \eps, z \right)$ and $P > 0$ on $\left( z, z + \eps \right)$. 
Then $\p \left( \lim_{n \to \infty} a_n = z \right) = 0$. 
Similarly, if $0 \in Z_P$ and $P > 0$ on $\left( 0, \eps \right)$, or if $1 \in Z_P$ and $P < 0$ on $\left( 1 - \eps, 1 \right)$, then the probability of convergence of $a_n$ to $0$ or $1$, respectively, is zero. 
\end{theorem}

\begin{theorem}[Nonlinear models, touchpoints]\label{thm:nonlin_touchpoints}
 Suppose that $p_k \neq k/m$ for some $0 \leq k \leq m$, and that $X_0, Y_0 > 0$. 
 Suppose that $z \in Z_P \cap \left( 0, 1 \right)$ is such that there exists an $\eps > 0$ such that $P$ is either strictly positive or strictly negative on the union of the intervals $\left( z -  \eps, z \right)$ and $\left( z, z + \eps \right)$. Then $\p \left( \lim_{n \to \infty} a_n = z \right) > 0$.
\end{theorem}

See Figure~\ref{fig:polynomial} for an illustration of the polynomial $P$ for various values of the parameters $\left\{ p_k \right\}_{0 \leq k \leq m}$, and what the various limiting proportions can be in each case.

\begin{figure}[h!]%[ht]
    \centering
    \begin{subfigure}[h]{0.3\textwidth}
	\centering
	\includegraphics[width=\textwidth]{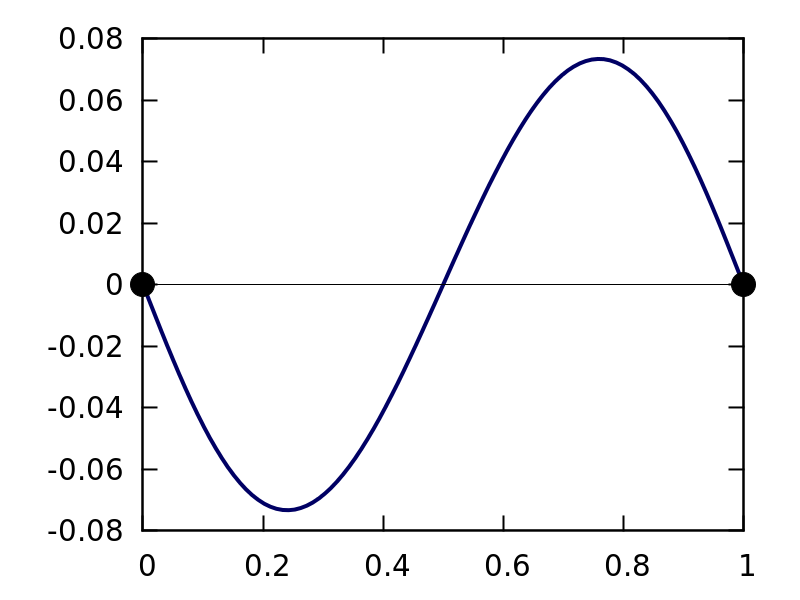}
	\caption{}
	\label{fig:poly1}
    \end{subfigure}
    \ \ 
    \begin{subfigure}[h]{0.3\textwidth}
	\centering
	\includegraphics[width=\textwidth]{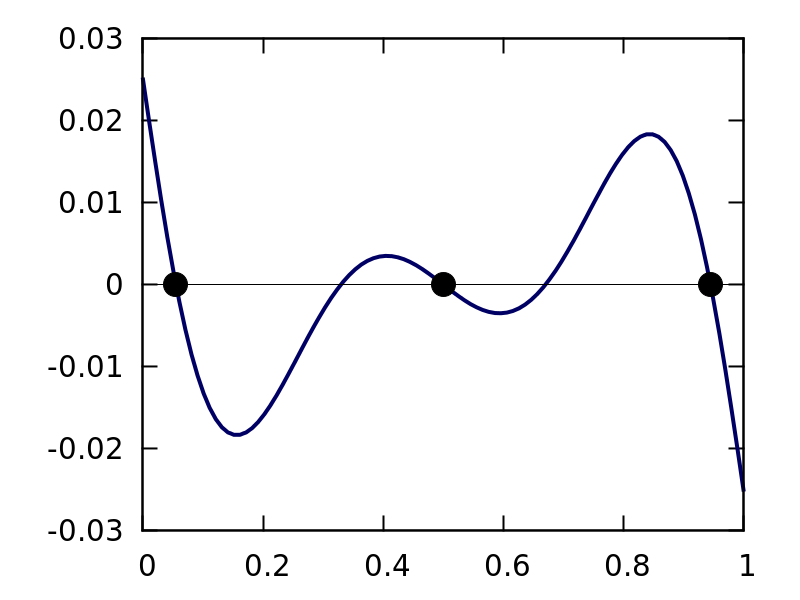}
	\caption{}
	\label{fig:poly2}
    \end{subfigure}
    \ \ 
    \begin{subfigure}[h]{0.3\textwidth}
	\centering
	\includegraphics[width=\textwidth]{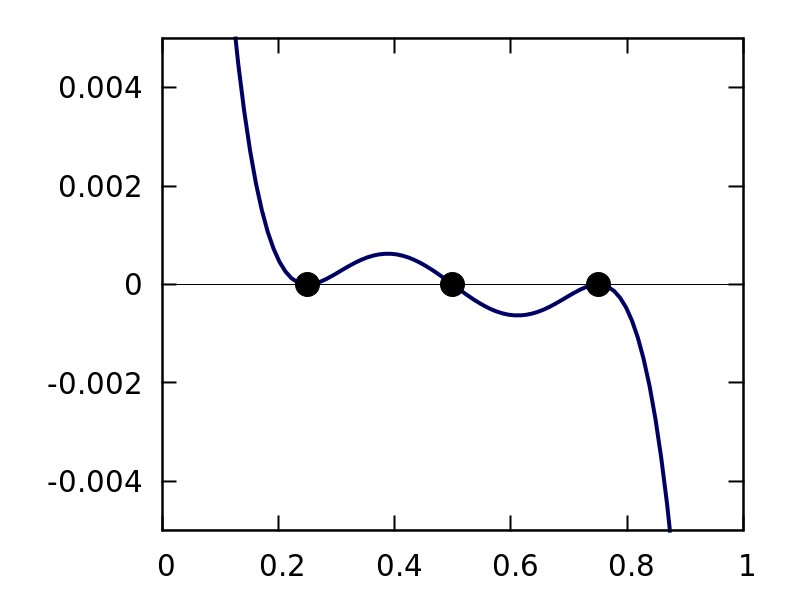}
	\caption{}
	\label{fig:poly3}
    \end{subfigure}
    \caption{\textbf{Examples of the polynomial $P$ and possible limiting proportions.} In each case there is no bias towards either colour, i.e., $p_k + p_{m-k} = 1$ for all $0 \leq k \leq m$. (a) Majority choice: $p_k = 1$ if $k > m/2$ and $p_k = 0$ otherwise ($m=5$ in the figure). The possible limits are 0 and 1, i.e., the winner takes all. (b)~The parameters here are: $m=9$, $p_5 = p_6 = 0.5$, $p_7 = p_8 = p_9 = 0.95$. Such an example is plausible if the strength of the signal from the neighbours matters: if $3 \leq k \leq 6$, then the signal towards either colour is weak, so just flip a fair coin to choose, but if $0 \leq k \leq 2$ or $7 \leq k \leq 9$ then there is a strong signal towards one of the colours, so pick that colour with probability close to 1. In this example the possible limits are: $z_1 \approx 0.055$, $z_2 = 0.5$, $z_3 \approx 0.945$, and there are also two zeros of $P$ which cannot be limits. (c) This is an example where $P$ has touchpoints. The parameters are: $m = 6$, $p_4 = 1031/
1710$, 
$p_5 = p_6 = 
35/38$, and the two touchpoints are at $z_1 = 1/4$ and $z_2 = 3/4$. Both of these, as well as $z_3 = 1/2$, can be limits.}
    \label{fig:polynomial}
\end{figure}

The theorems above provide a complete phase diagram of the asymptotic behaviour of the process in the case of two types. 
To illustrate this, see Figure~\ref{fig:phase}, which shows phase diagrams for $m=3$ and $m=4$ when there is no bias towards either colour, i.e., when $p_k + p_{m-k} = 1$ for all $0 \leq k \leq m$. This condition implies that $P\left( z \right) = - P \left( 1 - z \right)$ and so $1/2 \in Z_P$, but $1/2$ need not be a limit point (see Fig.~\ref{fig:polynomial}).

\begin{figure}[h!]%[ht]
    \centering
    \begin{subfigure}[h]{0.4\textwidth}
	\centering
	\includegraphics[width=\textwidth]{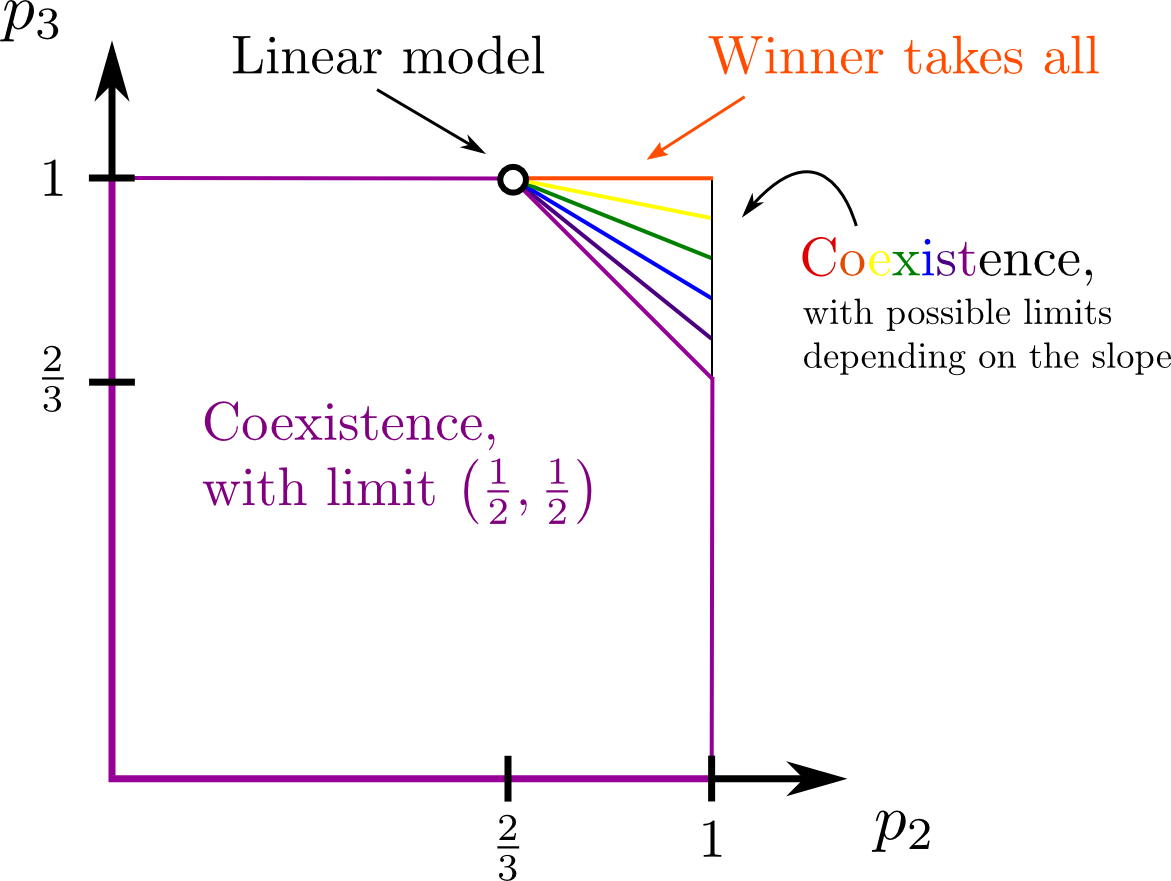}
	\caption{$m=3$}
	\label{fig:phase-m3}
    \end{subfigure}
    \qquad \qquad
    \begin{subfigure}[h]{0.4\textwidth}
	\centering
	\includegraphics[width=\textwidth]{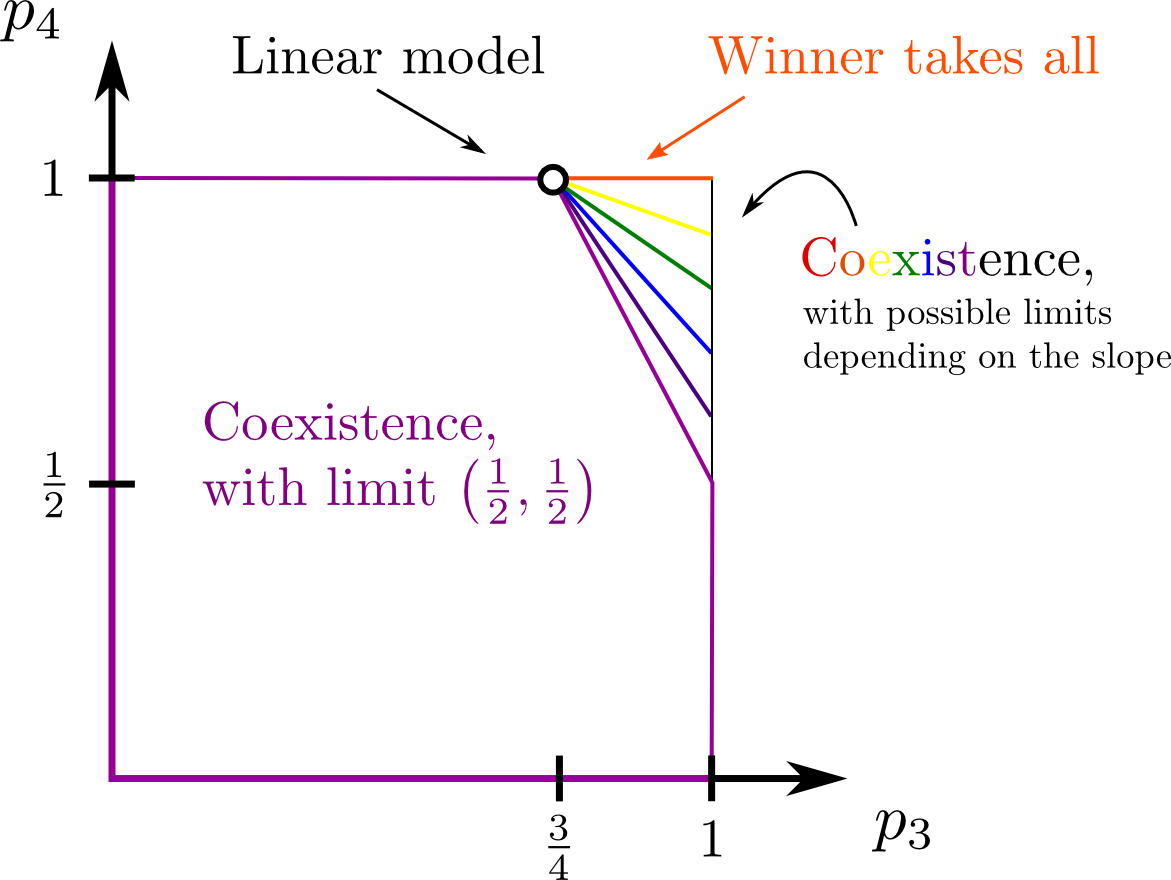}
	\caption{$m=4$}
	\label{fig:phase-m4}
    \end{subfigure}
    \caption{\textbf{Phase diagrams when there is no bias towards either colour/type}, i.e., when $p_k + p_{m-k} = 1$ for all $0 \leq k \leq m$. Let $q_k := p_k - k/m$. (a) If $q_2 < 0$ or if $q_2 + q_3 \leq 0$ and $q_3 < 0$, then $\lim_{n \to \infty} a_n = 1/2$, i.e., in this case the network is split evenly among the two types in the limit. The linear model is the case of $q_2 = q_3 = 0$. Finally, if $q_2 + q_3 > 0$ and $q_2 > 0$, then let $\alpha = - \frac{q_3}{q_2} \in [0,1)$; the possible limits of $a_n$ are then $\frac{1}{2} \pm \frac{1}{2} \sqrt{ \frac{3 - 3 \alpha}{3 + \alpha}}$. In particular, when $\alpha = 0$ then the winner takes all, and if $\alpha \in \left( 0, 1 \right)$, then the two types coexist in the limit. (b) This is similar to (a). Here, if $2q_3 + q_4 > 0$ and $q_4 > 0$, then let $\beta = - \frac{q_4}{q_3} \in [0,2)$; the possible limits of $a_n$ are then $\frac{1}{2} \pm \frac{1}{2} \sqrt{\frac{2 - \beta}{2 + \beta}}$.}
    \label{fig:phase}
\end{figure}

\textbf{Coexistence.} 
In particular, the theorems above show that in many cases the two colours coexist in the limit. 
Indeed, since $P(0) = \frac{1}{2}p_0$ and $P(1) = \frac{1}{2} \left( p_m - 1 \right)$, $p_0 = 0$ or $p_m = 1$ is necessary for one of the colours to asymptotically take over the network. 
Whenever $p_0 > 0$ and $p_m < 1$ the two colours coexist in the limit, and thus the model provides a theoretical understanding of coexistence in preferential attachment networks.

A natural extension of the model is to consider more than two colours. 
For clarity of presentation, we postpone the discussion of this until later: 
see Section~\ref{sec:many} for a description of the model with many colours and the corresponding results and conjectures.

It is also natural to consider variants of the model with different underlying network evolution models. 
For instance, consider the affine preferential attachment model, 
where a given edge connects the new node $v$ to a given existing node $u$ with probability proportional to the degree of $u$ plus some fixed constant $c$ which is greater than $-1$. 
This gives rise to a random graph with power-law degree distribution with exponent $3 + c$. 
The linear preferential attachment model is the special case with $c = 0$. 
The type adoption process with affine preferential attachment is slightly more involved, 
so to keep the exposition simple we do not go into details. 
However, the same techniques apply, 
and we get the same results as above for $c \geq 0$. 
For $c \in \left( - 1, 0 \right)$ the behavior is more delicate, 
and we leave it as an exercise to the reader to flesh out the details. 
See Section~\ref{sec:future} for further discussion on the affine preferential attachment model. 
Another natural variant is to consider uniform attachment---this is even simpler than preferential attachment, and the same results apply. 
We leave the understanding of other network evolution models for future work.

%%%%%%%%%%%%%%%%%%%%%%%%%%%%%%%%%
\subsection{Outline of the paper} %%%
%%%%%%%%%%%%%%%%%%%%%%%%%%%%%%%%%

First, in Section~\ref{sec:proofs} we prove the results described in Section~\ref{sec:results}. Then in Section~\ref{sec:many} we study the case of three or more types, and finally we conclude with open questions and directions for future research in Section~\ref{sec:future}.

%%%%%%%%%%%%%%%%%%%%%%%%%%%%%%%%%%%%%%
\section{Proofs}\label{sec:proofs} %%%
%%%%%%%%%%%%%%%%%%%%%%%%%%%%%%%%%%%%%%

This section contains the proofs of our main results described in Section~\ref{sec:results}, and is structured as follows. 
First, in Section~\ref{sec:reduction} we show how the asymptotic behaviour of $\left\{ a_n \right\}_{n \geq 0}$ is the same as that of the sum-of-degrees process $\left\{ x_n \right\}_{n \geq 0}$, which is more convenient to study, as it is a Markov process. 
Then in Section~\ref{sec:linear_proof} we study the linear model and prove Theorem~\ref{thm:continuous_case}.
Next in Section~\ref{sec:SAP} we recall results from the theory of stochastic approximation processes, 
and finally in Section~\ref{sec:nonlinear_proof} we prove our results concerning nonlinear models.

%%%%%%%%%%%%%%%%%%%%%%%%%%%%%%%%%%%%%%%%%%%%%%%%%%%%%%%%%%%%%%%%%%%%%%%%%%%%%
\subsection{Reduction to the sum-of-degrees process}\label{sec:reduction} %%%
%%%%%%%%%%%%%%%%%%%%%%%%%%%%%%%%%%%%%%%%%%%%%%%%%%%%%%%%%%%%%%%%%%%%%%%%%%%%%

To understand the process $\left\{ A_n \right\}_{n \geq 0}$ (and thus the normalised process $\left\{ a_n \right\}_{n \geq 0}$), it is more convenient to study the time evolution of the sum of the degrees of each type. 
The reason for this is that the process $\left\{ A_n \right\}_{n \geq 0}$ is not a Markov process, but the joint process $\left\{ \left( A_n, X_n \right) \right\}_{n \geq 0}$ is indeed Markov. 
It evolves as follows. Given $\left( A_n , X_n \right)$, $u_{n+1}$ is drawn from the binomial distribution with parameters $m$ and $x_n$. Subsequently, $I_{n+1}$ is drawn from the Bernoulli distribution with parameter $p_{u_{n+1}}$. We then have
\begin{align}
 A_{n+1} &= A_n + I_{n+1}, \label{eq:evol1}\\
 X_{n+1} &= X_n + u_{n+1} + m I_{n+1}. \label{eq:evol2}
\end{align}

The following lemma tells us that in order to understand the asymptotic behaviour of $\left\{ a_n \right\}_{n \geq 0}$, it is enough to understand the asymptotic behaviour of $\left\{ x_n \right\}_{n \geq 0}$. Consequently, in the following we analyse the latter, as this is a Markov process.

\begin{lemma}\label{lem:from_half_edges_to_vertices}
Suppose $\left\{ x_n \right\}_{n \geq 0}$ converges almost surely (a.s.) and let $x := \lim_{n \to \infty} x_n$ denote the limit. If $P\left( x \right) = 0$ a.s., then $\left\{ a_n \right\}_{n \geq 0}$ converges a.s.\ as well, and $\lim_{n \to \infty} a_n = x$ a.s.
\end{lemma}

\begin{proof}
 Let $\mcF_n$ denote the filtration of the process until time $n$. Given $\mcF_n$, the probability that the node added at time $n+1$ is red is
\begin{align*}
 \p \left( A_{n+1} - A_n  = 1 \, \middle| \, \mcF_n \right) &= \sum_{k=0}^m \binom{m}{k} x_n^k \left( 1 - x_n \right)^{m-k} p_k = x_n + \sum_{k=0}^m \binom{m}{k} x_n^k \left( 1 - x_n \right)^{m-k} q_k\\
 &= x_n + 2 P \left( x_n \right) =: f\left( x_n \right),
\end{align*}
where $q_k = p_k - k / m$. Thus $\E \left( A_{n+1} - A_n \, \middle| \, \mcF_n \right) = f \left( x_n \right)$. Define $M_n = A_n - A_0 - \sum_{i=0}^{n-1} f\left( x_i \right)$, with initial condition $M_0 = 0$. The previous calculation tells us that $\left\{ M_n \right\}_{n \geq 0}$ is a martingale with respect to the filtration $\mcF_n$. 
Moreover, this martingale has bounded increments since $M_{i+1} - M_i = A_{i+1} - A_i - f\left( x_i \right) \in \left[ - 1, 1 \right]$, and thus $\lim_{n \to \infty} M_n / n = 0$ a.s.

Let $x := \lim_{n \to \infty} x_n$. Since $P\left(x\right) = 0$, we have $f\left( x \right) = x$. Since $f$ is continuous, we have $\lim_{n \to \infty} f \left( x_n \right) = f\left( x \right) = x$ a.s., and thus the Ces\`{a}ro mean of the sequence $\left\{ f \left( x_n \right) \right\}_{n \geq 0}$ also converges to the same limit: $\lim_{n \to \infty} \frac{1}{n} \sum_{i=0}^{n-1} f\left( x_i \right) = x$ a.s. The claim then follows from the fact that $M_n / n = \frac{A_n}{n} - \frac{A_0}{n} - \frac{1}{n} \sum_{i=0}^{n-1} f\left( x_i \right)$ and $\lim_{n \to \infty} \left( a_n - \frac{A_n}{n} \right) = 0$.
\end{proof}

%%%%%%%%%%%%%%%%%%%%%%%%%%%%%%%%%%%%%%%%%%%%%%%%%%%%%
\subsection{Linear model}\label{sec:linear_proof} %%%
%%%%%%%%%%%%%%%%%%%%%%%%%%%%%%%%%%%%%%%%%%%%%%%%%%%%%

\begin{proof}[Proof of Theorem~\ref{thm:continuous_case}]
 In the linear model, when $p_k = \frac{k}{m}$ for all $k = 0, 1, \dots, m$, we have that $P \equiv 0$, and thus $\E \left( X_{n+1} - X_n \, \middle| \, \mcF_n \right) = 2m x_n$. Since $x_{n+1} - x_n = \frac{X_{n+1} - X_n - 2m x_n}{S_{n+1}}$, it follows that $\E \left( x_{n+1} - x_n \, \middle| \, \mcF_n \right) = 0$, i.e., $\left\{ x_n \right\}_{n \geq 0}$ is a martingale. Since it is also bounded, it converges almost surely. Lemma~\ref{lem:from_half_edges_to_vertices} then implies that $\left\{ a_n \right\}_{n \geq 0}$ converges a.s.\ as well, and $\lim_{n \to \infty} a_n = \lim_{n \to \infty} x_n$ a.s.

 We use a variance argument to show that the distribution of $x := \lim_{n \to \infty} x_n$ has full support on $\left[0,1\right]$. First note that $\left( x_{n+1} - x_n \right)^2 = \left( \frac{X_{n+1} - X_n - 2m x_n}{S_0 + 2m \left( n + 1 \right)} \right)^2 \leq \frac{1}{ \left( n + 1 \right)^2 }$, and consequently for any $n_0$ we have
\begin{equation}\label{eq:var_rem}
 \E \left( \left( x - x_{n_0} \right)^2 \, \middle| \, \mcF_{n_0} \right) = \lim_{n \to \infty} \E \left( \left( x_n - x_{n_0} \right)^2 \, \middle| \, \mcF_{n_0} \right) = \sum_{j = n_0}^{\infty} \E \left( \left( x_{j+1} - x_j \right)^2 \, \middle| \, \mcF_{n_0}  \right) \leq \sum_{j = n_0}^{\infty} \frac{1}{\left( j + 1 \right)^2} \leq \frac{1}{n_0}.
\end{equation}
 Now let $\left( r, r + \eps \right) \subset \left( 0, 1 \right)$ be any fixed interval. 
Our goal is to show that $\p \left( x \in \left( r, r + \eps \right) \right) > 0$. 
Let $n_0$ be an integer such that $n_0 \geq \frac{18}{\eps^2}$ and $\p \left( x_{n_0} \in \left( r + \frac{\eps}{3}, r + \frac{2\eps}{3} \right) \right) > 0$ (this is possible since for large enough $n_0$ there exists a sequence of events such that $x_{n_0} \in \left( r + \frac{\eps}{3}, r + \frac{2\eps}{3} \right)$). 
Now condition on this event;~\eqref{eq:var_rem} implies that 
\[
 \E \left( \left( x - x_{n_0} \right)^2 \, \middle| \, x_{n_0} \in \left( r + \frac{\eps}{3}, r + \frac{2\eps}{3} \right) \right) \leq \frac{1}{n_0} \leq \frac{\eps^2}{18},
\]
which in turn implies that $\p \left( \left| x - x_{n_0} \right| \leq \frac{\eps}{3} \, \middle| \,  x_{n_0} \in \left( r + \frac{\eps}{3}, r + \frac{2\eps}{3} \right)  \right) \geq \frac{1}{2}$. We conclude that
\[
 \p \left( x \in \left( r, r + \eps \right) \right) \geq \p \left( \left| x - x_{n_0} \right| \leq \frac{\eps}{3} \, \middle| \,  x_{n_0} \in \left( r + \frac{\eps}{3}, r + \frac{2\eps}{3} \right)  \right) \p \left( x_{n_0} \in \left( r + \frac{\eps}{3}, r + \frac{2\eps}{3} \right) \right) > 0.
\]

Finally, showing that the distribution of $x$ has no atoms can be done by adapting arguments by Pemantle~\cite{pemantle1990time}. First, let us describe how the process $\left\{ x_n \right\}_{n \geq 0}$ is related to time-dependent P\'olya urn processes that Pemantle studies in~\cite{pemantle1990time}.

Time-dependent P\'olya urn processes are generalisations of the classical P\'olya urn process, where the number of balls added to the urn is allowed to vary with time. Although $\left\{ x_n \right\}_{n \geq 0}$ is not a time-dependent P\'olya urn process, the following slight modification of the preferential attachment process does give a time-dependent P\'olya urn process. When adding a new  node $v$ to the graph $G_n = \left( V_n, E_n \right)$, add its $m$ neighbours one by one, and after adding each neighbour, \emph{update} the degree of the neighbour. Let $\widetilde{X}_n$ denote the sum of the degrees of red nodes at time $n$ in this model. Consider also a time-dependent P\'olya urn process $\left\{ Z_n \right\}_{n \geq 0}$ where at times $t \neq 0 \mod m$ a single ball is added to the urn, and at times $t = 0 \mod m$ the number of balls added to the urn is $m+1$. It can be seen that if $\widetilde{X}_0 = Z_0$, then $\widetilde{X}_n$ and $Z_{mn}$ have the same distribution. Thus Pemantle's results~\cite[
Theorem~3, Theorem~4]{
pemantle1990time} apply directly and show that the distribution of $\lim_{n \to \infty} \widetilde{x}_n$ (this limit exists a.s.) has no atoms.

Since our setting is close to Pemantle's original setting, we only sketch the proof that the distribution of $x$ has no atoms, and leave the details to the reader.

To show that the distribution of $x$ has no atoms on $\left( 0, 1 \right)$, we can adapt the variance arguments of~\cite[Theorem~3]{pemantle1990time}. Fix $r \in \left( 0, 1 \right)$. Suppose on the contrary that $\p \left( x = r \right) > 0$. Then for every $\eps > 0$ there exists $n_0$ and some event $\cA \in \cF_{n_0}$ having positive probability such that $\p \left( x_n \to r \, \middle| \, \cA \right) \geq 1 - \eps$; in fact, $n_0$ can be as large as desired. Define $c := \frac{ r \left( 1 - r \right)^{m/2}}{10 \times 2^{m/2}}$ and let $N = \max \left\{ \frac{S_0}{m}, \frac{2}{c^2 \min \left\{ r, 1 - r \right\}} \right\}$. One can then show, via variance arguments, the following two inequalities. First, for every $n \geq N$,
\[
 \p \left( \sup_{k \geq n} \left| x_k - r \right| \geq \frac{c}{\sqrt{n}} \, \middle| \, \cF_n \right) \geq \frac{1}{2}.
\]
Second, defining $\cB = \left\{ \left| x_n - r \right| \geq \frac{c}{\sqrt{n}} \right\}$, we have that for every $n \geq N$,
\[
 \p \left( \inf_{k \geq n} \left| x_k - r \right| \geq \frac{c}{2\sqrt{n}} \, \middle| \, \cF_n, \cB \right) \geq \frac{c^2}{16}.
\]
Putting these together we have that for every $n \geq N$, the probability given $\cF_n$ is at least $\frac{c^2}{32}$ that some $x_{n+k}$ will be at least $\frac{c}{\sqrt{n}}$ away from $r$ and no subsequent $x_{n+k+\ell}$ will ever return to the interval $\left[ r - \frac{c}{2\sqrt{n}}, r+ \frac{c}{2\sqrt{n}} \right]$. This contradicts our initial assumption and so $\p \left( x = r \right) = 0$.

To show that the distribution of $x$ has no atoms at $0$ and $1$, we can adapt the arguments of~\cite[Theorem~4]{pemantle1990time}. 
The main idea is a domination argument. 
Let $\left\{ v_n \right\}_{n \geq 0}$ be the P\'olya urn process where at each time step $2m$ balls are added to the urn, and let $v_0 = x_0$. 
Then the distribution of $x_n$ can be dominated by the distribution of $v_n$, in the sense that $\E \left( h \left( x_n \right) \right) \leq \E \left( h \left( v_n \right) \right)$ for every continuous bounded convex function $h$. 
In other words, $x_n$ is smaller than $v_n$ in the convex order~\cite{shaked2007stochastic}. 
Since the limiting distribution of $\left\{ v_n \right\}_{n \geq 0}$ is a beta distribution, which does not have an atom at zero, one can then take $h_\eps \left( x \right) := \max \left\{ 0, 2 - x / \eps \right\}$ and let $\eps \to 0$ to conclude that the distribution of $x$ cannot have an atom at zero either. 
We refer the reader to~\cite[Theorem~4]{pemantle1990time} for more details. 
See also the proof of Theorem~\ref{thm:nonlin_unstable} for the endpoints in Section~\ref{sec:nonlinear_proof}. 
\end{proof}

%%%%%%%%%%%%%%%%%%%%%%%%%%%%%%%%%%%%%%%%%%%%%%%%%%%%%%%%%%%%%%%%%%
\subsection{Stochastic approximation processes}\label{sec:SAP} %%%
%%%%%%%%%%%%%%%%%%%%%%%%%%%%%%%%%%%%%%%%%%%%%%%%%%%%%%%%%%%%%%%%%%

The key observation in the analysis of the asymptotic behaviour of $\left\{ x_n \right\}_{n \geq 0}$ is that it is a stochastic approximation process. 
Stochastic approximation was introduced in 1951 by Robbins and Monro~\cite{robbins1951stochastic}, whose goal was to approximate the root of an unknown function via evaluation queries that are necessarily noisy. There has been much follow-up research, see, e.g., the monograph by Nevelson and Hasminskii~\cite{nevelson1976stochastic}. The setup of stochastic approximation arises naturally in the study of P\'olya urn processes; see the survey~\cite{pemantle2007survey} for details. In particular, we use results of Hill, Lane and Sudderth~\cite{hill1980strong}, who studied generalised (nonlinear) P\'olya urn processes, and we also use subsequent refinements by Pemantle~\cite{pemantle1990nonconvergence,pemantle1991touchpoints}. We state the main theorems here and refer to the original papers for more details; see also the survey~\cite{pemantle2007survey}. Stochastic approximation results in higher dimensions will be discussed in Section~\ref{sec:many}.

Let $\left\{ Z_n \right\}_{n \geq 0}$ be a stochastic process in $\R$ adapted to a filtration $\left\{ \mcF_n \right\}$. Suppose that it satisfies
\begin{equation}\label{eq:stoch_appx}
 Z_{n+1} - Z_n = \frac{1}{n} \left( F \left( Z_n \right) + \xi_{n+1} + R_n \right),
\end{equation}
where $F : \R \to \R$, $\E \left( \xi_{n+1} \, \middle| \, \mcF_n \right) = 0$, and the remainder terms $R_n \in \mcF_n$ go to zero and also satisfy $\sum_{n=1}^{\infty} n^{-1} \left| R_n \right| < \infty$ almost surely. Such a process is known as a \emph{stochastic approximation process}.

Intuitively, trajectories of a stochastic approximation process $\left\{ Z_n \right\}_{n \geq 0}$ should approximate the trajectories $\left\{ Z \left( t \right) \right\}_{t \geq 0}$ of the corresponding ODE $dZ / dt = F \left( Z \right)$. Moreover, since $\left\{ Z_n \right\}_{n \geq 0}$ is a stochastic system, we expect that stable trajectories of the ODE should appear, but unstable trajectories should not. This intuition is confirmed and formalized in the following statements (quoted from the survey~\cite{pemantle2007survey}); for proofs and more details see the papers cited above.

\begin{theorem}[Convergence to the zero set of $F$]\label{thm:conv_zero_set_F}
Suppose that $\left\{ Z_n \right\}$ is a stochastic approximation process and that $\E \left( \xi_{n+1}^2 \, \middle| \, \mcF_n \right) \leq K$ for some finite constant $K$. 
If $F$ is bounded and continuous, then $Z_n$ converges almost surely to the zero set of $F$.
\end{theorem}

\begin{theorem}[Convergence to stable equilibria]\label{thm:conv_stable_fix_pts}
Suppose that $\left\{ Z_n \right\}$ is a stochastic approximation process with a bounded and continuous $F$, and that $\E \left( \xi_{n+1}^2 \, \middle| \, \mcF_n \right) \leq K$ for some finite constant $K$. 
Suppose that there is a point $z$ and an $\eps > 0$ with $F \left( z \right) = 0$, $F > 0$ on $\left( z - \eps, z \right)$ and $F < 0$ on $\left( z, z + \eps \right)$. Then $\p \left( Z_n \to z \right) > 0$. Similarly, when $F : \left[0,1\right] \to \R$, if $F\left(0\right) = 0$ and $F < 0$ on $\left( 0, \eps \right)$ or if $F\left(1 \right) = 0$ and $F > 0$ on $\left( 1 - \eps, 1 \right)$, then there is a positive probability of convergence to 0 or 1, respectively.
\end{theorem}

\begin{theorem}[Nonconvergence to unstable equilibria]\label{thm:nonconv_unstable_fix_pts}
 Suppose that $\left\{ Z_n \right\}$ is a stochastic approximation process with a bounded and continuous $F$. 
Suppose that there is a point $z \in \left( 0, 1 \right)$ and an $\eps > 0$ with $F \left( z \right) = 0$, $F < 0$ on $\left( z - \eps, z \right)$ and $F > 0$ on $\left( z, z + \eps \right)$. 
Let $x^{+} = \max \left\{ x, 0 \right\}$ and $x^{-} = - \min \left\{ x, 0 \right\}$ denote the positive and negative parts of $x$, respectively. 
Suppose further that when $Z_n \in \left( z - \eps, z + \eps \right)$ then 
$\E \left( \xi_{n+1}^{+} \, \middle| \, \mcF_n \right)$ and 
$\E \left( \xi_{n+1}^{-} \, \middle| \, \mcF_n \right)$ 
are bounded above and below by positive constants depending only on $\eps$. 
Then $\p \left( Z_n \to z \right) = 0$. 
\end{theorem}

Pemantle studied the case of touchpoints for generalised (nonlinear) P\'olya urn processes in~\cite{pemantle1991touchpoints}. His proof extends to the following result.

\begin{theorem}[Convergence to touchpoints]\label{thm:conv_touchpoints}
 Suppose that $\left\{ Z_n \right\}$ is a stochastic approximation process with a bounded and continuously differentiable $F$, and that $\left| \xi_n \right| \leq K$ a.s.\ for some finite constant $K$. 
 Suppose that $z \in Z_P$ is a touchpoint, i.e., there exists an $\eps > 0$ such that either $F > 0$ on $\left( z - \eps, z \right) \cup \left( z , z + \eps \right)$ or $F < 0$ on $\left( z - \eps, z \right) \cup \left( z , z + \eps \right)$. 
Then $\p \left( Z_n \to z \right) > 0$.
\end{theorem}

%%%%%%%%%%%%%%%%%%%%%%%%%%%%%%%%%%%%%%%%%%%%%%%%%%%%%%%%%%%%
\subsection{Nonlinear models}\label{sec:nonlinear_proof} %%%
%%%%%%%%%%%%%%%%%%%%%%%%%%%%%%%%%%%%%%%%%%%%%%%%%%%%%%%%%%%%

We first show that $\left\{ x_n \right\}_{n \geq 0}$ is a stochastic approximation process (i.e., that it  satisfies~\eqref{eq:stoch_appx}) with the function $P$ as in~\eqref{eq:P}. Subsequently we show how this implies our results in Section~\ref{sec:results} using the results described in Section~\ref{sec:SAP}.

\begin{lemma}\label{lem:stoch_appx}
 The process $\left\{ x_n \right\}_{n \geq 0}$ is a stochastic approximation process with the function $F = P$ as in~\eqref{eq:P}. Furthermore, the noise term $\xi_n$ is bounded: $\left| \xi_n \right| \leq 2$ for all $n \geq 1$.
\end{lemma}

\begin{proof}
From~\eqref{eq:evol2} we have that the conditional expectation of $X_{n+1} - X_n$ is:
\[
 \E \left( X_{n+1} - X_n \, \middle| \, \mcF_n \right) = \sum_{k=0}^m \binom{m}{k} x_n^k \left( 1 - x_n \right)^{m-k} \left( k + m p_k \right) = 2 m x_n + 2m P \left( x_n \right).
\]
One can check that $x_{n+1} - x_n = \frac{X_{n+1} - X_n - 2m x_n}{S_{n+1}}$ and consequently $\E \left( x_{n+1} - x_n \, \middle| \, \mcF_n \right) = \frac{2m}{S_{n+1}} P \left( x_n \right)$, 
with $P$ as in~\eqref{eq:P}. 
We can then write $\left\{ x_n \right\}_{n \geq 0}$ as a stochastic approximation process as claimed in the statement of the lemma, i.e., we can write
\[
 x_{n+1} - x_n = \frac{1}{n} \left( P \left( x_n \right) + \xi_{n+1} + R_n \right)
\]
with appropriately defined $\xi_{n+1}$ and $R_n$. Define $\xi_{n+1}$ as
\begin{equation}\label{eq:noise}
 \xi_{n+1} = n \left( x_{n+1} - x_n - \E \left( x_{n+1} - x_n \, \middle| \, \mcF_n \right) \right).
\end{equation}
The remainder term $R_n$ can then be written as
\[
 R_n = - \frac{S_0 + 2m}{S_0 + 2m \left( n + 1 \right)} P \left( x_n \right).
\]

Clearly $R_n \in \mcF_n$. Let us now show that $\sum_{n=1}^{\infty} n^{-1} \left| R_n \right| < \infty$. A crude bound on $P$ is $\left| P \left( t \right) \right| \leq \frac{1}{2} \sum_{k=0}^m \binom{m}{k} \left| p_k - k/m \right| t^k \left( 1 - t \right)^{m-k} \leq \frac{1}{2} \sum_{k=0}^m \binom{m}{k} t^k \left( 1 - t \right)^{m-k} = \frac{1}{2}$. Therefore $\left| R_n \right| \leq \frac{1}{2} \frac{S_0 + 2m}{S_0 + 2m \left( n + 1 \right)}$, so indeed we have $\sum_{n=1}^{\infty} n^{-1} \left| R_n \right| < \infty$.

Finally, to bound the noise term, notice that $\left| x_{n+1} - x_n \right| = \left| \frac{X_{n+1} - X_n - 2m x_n}{S_0 + 2m \left( n + 1 \right)} \right| \leq \frac{2m}{2m \left( n + 1 \right)} = \frac{1}{n+1}$. Then using~\eqref{eq:noise} and the triangle inequality, we get that $\left| \xi_n \right| \leq 2$.
\end{proof}

The results in Section~\ref{sec:results} now follow. First, note that Lemma~\ref{lem:from_half_edges_to_vertices} implies that it is enough to show the claims in Theorems~\ref{thm:point_mass_case},~\ref{thm:nonlin_stable},~\ref{thm:nonlin_unstable}, and~\ref{thm:nonlin_touchpoints} for the process $\left\{ x_n \right\}_{n \geq 0}$ (instead of for the process $\left\{ a_n \right\}_{n \geq 0}$).

\begin{proof}[Proof of Theorem~\ref{thm:point_mass_case}]
 This follows directly from Lemma~\ref{lem:stoch_appx} and Theorem~\ref{thm:conv_zero_set_F}.
\end{proof}

\begin{proof}[Proof of Theorem~\ref{thm:nonlin_stable}]
 This follows directly from Lemma~\ref{lem:stoch_appx} and Theorem~\ref{thm:conv_stable_fix_pts}.
\end{proof}

The proof of Theorem~\ref{thm:nonlin_unstable} is more involved. This is in line with related work in the literature, where conditions for nonconvergence to unstable equilibria are more difficult to find than similar results for convergence to stable equilibria (see~\cite{pemantle2007survey} for a discussion). Recall the proof of Theorem~\ref{thm:continuous_case}, where we showed that the limiting distribution in the linear model has no atoms: we used a variance argument for points in $\left( 0, 1 \right)$, and a domination argument for the endpoints $0$ and $1$. Our proof of Theorem~\ref{thm:nonlin_unstable} follows similar lines.

We first proceed by proving Theorem~\ref{thm:nonlin_unstable} for points $z \in \left( 0, 1 \right) \cap Z_P$. 
Intuitively, the process has sufficient noise which prevents it from converging to $z$. 
The following lemma is key to bounding the noise of the process from below. 
Its proof is simple when $p_0 < 1$ and $p_m > 0$; 
the proof is only lengthy because it deals with the special cases when $p_0 = 1$ or $p_m = 0$.

\begin{lemma}\label{lem:lower}
 Suppose that the parameters $\left\{ p_k \right\}_{0 \leq k \leq m}$ do not fall into one of the following three cases: (a) $p_k = 0$ for all $0 \leq k \leq m$; (b) $p_k = 1$ for all $0 \leq k \leq m$; (c) $m=1$, $p_0 = 1$, $p_1 = 0$. Suppose that $z \in \left( 0, 1 \right) \cap Z_P$. Then there exist integers $k_1$ and $k_2$ such that $k_1 < 2m z < k_2$ and, if $x_n \in \left( \delta, 1 - \delta \right)$ for some $\delta > 0$, then the probabilities $\p \left( X_{n+1} - X_n = k_1 \, \middle| \, \mcF_n \right)$ and $\p \left( X_{n+1} - X_n = k_2 \, \middle| \, \mcF_n \right)$ are bounded away from zero by a positive function of $\delta$ and the parameters $\left\{ p_k \right\}_{0 \leq k \leq m}$.
\end{lemma}

\begin{proof}
In the following we always assume that $x_n \in \left( \delta, 1 - \delta \right)$. If $p_0 < 1$ then we can choose $k_1 = 0$ since $\p \left( X_{n+1} - X_n = 0 \, \middle| \, \mcF_n \right) = \left( 1 - x_n \right)^m \left( 1 - p_0 \right) \geq \delta^m \left( 1 - p_0 \right)$. Similarly, if $p_m > 0$ then we can choose $k_2 = 2m$, since $\p \left( X_{n+1} - X_n = 2m \, \middle| \, \mcF_n \right) = x_n^m p_m \geq \delta^m p_m$. The rest of the proof deals with the cases when either $p_0 = 1$ or $p_m = 0$.
 
First consider the case when $p_0 > 0$ and $p_1 = p_2 = \dots = p_m = 0$. In this case $P \left( s \right) = \frac{1}{2} \left[ p_0 \left( 1 - s \right)^m - s \right]$, which is decreasing in $\left[0,1\right]$, so it has a single zero in $\left(0,1\right)$. 
In fact, $P\left( 1/2 \right) < 0$, so the single zero of $P$ in $\left( 0, 1 \right)$ is in $\left( 0, 1/2 \right)$, and thus we can take $k_2 = m$.
If $p_0 < 1$ then we can take $k_1 = 0$ as described above. 
Finally, if $p_0 = 1$ and $m > 2$, then we can take $k_1 = 1$. 
This is because the zero of $P$ in $\left( 0, 1 \right)$ is in $\left( \frac{1}{2m}, \frac{1}{2} \right)$, which follows from the fact that $P \left( \frac{1}{2m} \right) > 0$. The case when $p_0 = p_1 = \dots = p_{m-1} = 1$ and $p_m < 1$ follows similarly.

Now we can assume that there exist $1 \leq i \leq m$ and $0 \leq j \leq m - 1$ such that $p_i > 0$ and $p_j < 1$. This implies that $\p \left( X_{n+1} - X_n = j \, \middle| \, \mcF_n \right) \geq \delta^m \left( 1 - p_j \right) > 0$ and $\p \left( X_{n+1} - X_n = m+i \, \middle| \, \mcF_n \right) \geq \delta^m p_i > 0$. Thus if $z = 1/2$ then we can take $k_1 = j$ and $k_2 = m + i$. If $0 < z < 1/2$ then we can again take $k_2 = m+i$, and we just need to show the existence of an appropriate $k_1$. Assume by contradiction that there does not exist an appropriate $k_1$, i.e., for all $\ell < 2m z$, $p_\ell = 1$. Then we have 
\[
 P \left( s \right) \geq \frac{1}{2} \left[ \sum_{0 \leq k < 2m z} \binom{m}{k} s^k \left( 1 - s \right)^{m-k} - s \right] = \frac{1}{2} \left[ 1 - s - \sum_{2m z \leq k \leq m} \binom{m}{k} s^k \left( 1 - s \right)^{m-k} \right].
\]
By Markov's inequality for a binomial random variable, this latter sum evaluated at $z$ is at most $1/2$, and since $z < 1/2$, we must have $P \left( z \right) > 0$, which is a contradiction. The case of $1/2 < z < 1$ is similar.
\end{proof}

\begin{proof}[Proof of Theorem~\ref{thm:nonlin_unstable} for $z \in \left( 0, 1 \right)$]
 This follows from Lemma~\ref{lem:stoch_appx} and Theorem~\ref{thm:nonconv_unstable_fix_pts}. The only condition of Theorem~\ref{thm:nonconv_unstable_fix_pts} that needs to be checked additionally is that  $\E \left( \xi_{n+1}^{+} \, \middle| \, \mcF_n \right)$ and $\E \left( \xi_{n+1}^{-} \, \middle| \, \mcF_n \right)$ are bounded away from zero by positive numbers when $x_n \in \left( z - \eps, z + \eps \right)$ for small enough $\eps > 0$; this can be done using Lemma~\ref{lem:lower}. In the special cases (a), (b), and (c) described in Lemma~\ref{lem:lower}, the statement of Theorem~\ref{thm:nonlin_unstable} is vacuously true, since in each case the polynomial $P$ has no zeros at which it is increasing. Thus we may assume that we are not in these special cases, and we can use Lemma~\ref{lem:lower}. Recall that
\[
 \xi_{n+1} = \frac{n}{S_{n+1}} \left\{ X_{n+1} - X_n - 2m \left( x_n + P \left( x_n \right) \right) \right\}.
\]
Define $\widetilde{\eps} := \frac{1}{2} \min \left\{ 2mz - k_1, k_2 - 2mz \right\}$, where $k_1$ and $k_2$ are given by Lemma~\ref{lem:lower}, and let $\eps > 0$ be small enough such that whenever $x_n \in \left( z - \eps, z + \eps \right)$, necessarily $2m \left( x_n + P \left( x_n \right) \right) \in \left( 2mz - \widetilde{\eps}, 2mz + \widetilde{\eps} \right)$. 
If $x_{n} \in \left( z - \eps, z + \eps \right)$ then we have $\E \left( \xi_{n+1}^+ \, \middle| \, \mcF_n \right) \geq \frac{n}{S_{n+1}} \widetilde{\eps} \p \left( X_{n+1} - X_n = k_2 \, \middle| \, \mcF_n \right)$, where $\lim_{n \to \infty} \frac{n}{S_{n+1}} = \frac{1}{2m}$, and by Lemma~\ref{lem:lower} the probability $\p \left( X_{n+1} - X_n = k_2 \, \middle| \, \mcF_n \right)$ is bounded from below by a positive function of $z$, $\eps$, and the parameters $\left\{ p_k \right\}_{0 \leq k \leq m}$. 
We can similarly bound $\E \left( \xi_{n+1}^- \, \middle| \, \mcF_n \right)$ from below.
\end{proof}

We next prove Theorem~\ref{thm:nonlin_unstable} for the endpoints $0$ and $1$. 
The main idea of the proof is to compare the behaviour near the endpoints of our process of interest to that of a standard P\'olya urn process where $2m$ balls are added at each time step. 
In order to formalize this, we make use of several different stochastic orders; we refer to~\cite{shaked2007stochastic} for an overview of these. 
We proceed by defining these stochastic orders and stating a few results on them, before proving Theorem~\ref{thm:nonlin_unstable}.

\begin{definition}[Stochastic orders]\label{def:stochastic_orders}
 Let $X$ and $Y$ be random variables.

 We say that $X$ is \emph{smaller than $Y$ in the usual stochastic order} (denoted by $X \leq_{\st} Y$) if $\E \left( \phi \left( X \right) \right) \leq \E \left( \phi \left( Y \right) \right)$ for all increasing continuous functions $\phi : \R \to \R$ for which these expectations exist.

 We say that $X$ is \emph{smaller than $Y$ in the convex order} (denoted by $X \leq_{\cx} Y$) if $\E \left( \phi \left( X \right) \right) \leq \E \left( \phi \left( Y \right) \right)$ for all  continuous convex functions $\phi : \R \to \R$ for which these expectations exist.

 We say that $X$ is \emph{smaller than $Y$ in the increasing convex order} (denoted by $X \leq_{\icx} Y$) if $\E \left( \phi \left( X \right) \right) \leq \E \left( \phi \left( Y \right) \right)$ for all increasing continuous convex functions $\phi : \R \to \R$ for which these expectations exist.
\end{definition}

\begin{lemma}\label{lem:st-cx-icx}
 Two random variables $X$ and $Y$ satisfy $X \leq_{\icx} Y$ if and only if there is a random variable $Z$ such that $X \leq_{\st} Z \leq_{\cx} Y$.
\end{lemma}

\begin{proof}
 See~\cite[Theorem~4.A.6.~(a)]{shaked2007stochastic}.
\end{proof}

\begin{lemma}\label{lem:icx_suff}
 Let $X$ and $Y$ be two random variables with cumulative distribution functions $F$ and $G$, respectively, and bounded supports. Suppose that $\E \left( X \right) \leq \E \left( Y \right)$, and also that if $t_1 < t_2$ and $G\left( t_1 \right) < F \left( t_1 \right)$ then $G\left( t_2 \right) \leq F\left( t_2 \right)$. Then $X \leq_{\icx} Y$.
\end{lemma}

\begin{proof}
 See~\cite[Theorem~4.A.22.~(b)]{shaked2007stochastic}.
\end{proof}

\begin{lemma}\label{lem:Polya-cx}
 Consider the standard P\'olya urn process where $2m$ balls are added at each time step. Let $x_n^1$ and $x_n^2$ be the proportions of red balls at the $n^{\text{th}}$ step of two realizations of this process. If $x_n^1 \leq_{\cx} x_n^2$, then $x_{n+1}^1 \leq_{\cx} x_{n+1}^2$, i.e., the P\'olya urn process preserves dominance in the convex order.
\end{lemma}

\begin{proof}
 See Proposition~1 in~\cite{pemantle1990time}, in particular equation~(13).
\end{proof}

\begin{proof}[Proof of Theorem~\ref{thm:nonlin_unstable} for the endpoints]
We prove nonconvergence to 1 when $P\left( 1 \right) = 0$ and $P < 0$ on $\left( 1 - \eps, 1 \right)$ for some $\eps > 0$; 
the proof for the other endpoint is analogous. 
In the following fix $0 < \eps < 1/m$. 

The main idea of the proof is to compare the behaviour near $1$ of our process of interest to that of a standard P\'olya urn process where $2m$ balls are added at each time step. 
To aid in this comparison we also introduce an auxiliary process which is a combination of these two. We begin by describing these processes. 

Our process of interest is $\left\{ X_n \right\}_{n \geq 0}$, together with its normalised process $\left\{x_n \right\}_{n \geq 0}$. 
Let $\left\{ \overline{X}_n \right\}_{n \geq 0}$ denote the process of the number of red balls in a standard P\'olya urn process where $2m$ balls are added at each time step and the initial conditions are the same as those for the process $\left\{ X_n \right\}_{n \geq 0}$, i.e., $\overline{X}_0 = X_0$. 
Let $\left\{ \overline{x}_n \right\}_{n \geq 0}$ denote the normalised process, i.e., $\overline{x}_n = \frac{\overline{X}_n}{S_0 + 2mn}$. 
Let $\left\{ \widetilde{X}_n \right\}_{n \geq 0}$ denote the auxiliary process, with initial condition $\widetilde{X}_0 = X_0$, 
and let $\left\{ \widetilde{x}_n \right\}_{n \geq 0}$ denote the normalised process, i.e., $\widetilde{x}_n = \frac{\widetilde{X}_n}{S_0 + 2mn}$. 
We define this auxiliary process as follows. 
For $1 - \eps < x \leq 1$, given $\widetilde{x}_n = x$ let $\widetilde{X}_{n+1}$ have the same distribution as $X_{n+1}$ given $x_n = x$. 
For $x \leq 1 - \eps$, let 
\begin{align*}
\p \left( \widetilde{X}_{n+1} = \widetilde{X}_n \, \middle| \, \widetilde{x}_n = x \right) &= 1 - x \\
\text{ and } \quad 
\p \left( \widetilde{X}_{n+1} = \widetilde{X}_n + 2m \, \middle| \, \widetilde{x}_n = x \right) &= x. 
\end{align*}
In other words, 
when $\widetilde{x}_n > 1 - \eps$ then evolve the auxiliary process according to our process of interest, 
and when $\widetilde{x}_n \leq 1 - \eps$ then evolve it as a P\'olya urn process. 

We first show that it suffices to prove the claim for the auxiliary process, i.e., 
it suffices to show that $\p \left( \lim_{n \to \infty} \widetilde{x}_n = 1 \right) = 0$. 
Define the following events:
\[
 A_n := \left\{ \lim_{k \to \infty} x_k = 1, x_k > 1 - \eps \text{ for all } k \geq n \right\}, \qquad \qquad
 \widetilde{A}_n := \left\{ \lim_{k \to \infty} \widetilde{x}_k = 1, \widetilde{x}_k > 1 - \eps \text{ for all } k \geq n \right\}.
\]
If $\p \left( \lim_{n \to \infty} x_n = 1 \right) > 0$, 
then there exists $n_0 < \infty$ such that $\p \left( A_{n_0} \right) > 0$. 
In particular, there exists $y_0 \in \left( 1 - \eps, 1 \right)$ such that 
both probabilities $\p \left( x_{n_0} \geq y_0 \right)$ and $\p \left( A_{n_0} \, \middle| \, x_{n_0} = y_0 \right)$ are positive. 
In fact, we claim that $\p \left( A_{n_0} \, \middle| \, x_{n_0} = y \right)$ is positive for all $y_0 \leq y < 1$. 

To see this, consider two realizations of our process, 
$\left\{ X_n^1 \right\}_{n \geq 0}$ and $\left\{ X_n^2 \right\}_{n \geq 0}$, 
together with the normalised processes $\left\{ x_n^1 \right\}_{n \geq 0}$ and $\left\{ x_n^2 \right\}_{n \geq 0}$. 
Given $1 - \eps < x_n^1 < x_n^2 < 1$, 
we can couple $X_{n+1}^1$ and $X_{n+1}^2$ such that for any $0 \leq k \leq 2m$, 
$X_{n+1}^1 - X_n^1 = k$ implies that either $X_{n+1}^2 - X_n^2 = k$ or $X_{n+1}^2 - X_n^2 = 2m$. 
This is possible due to two facts. 
First, since $\eps < 1/m$, 
on the interval $\left( 1 - \eps, 1 \right)$ the function $x \mapsto x^m$ is increasing, 
while for $0 \leq k < m$, the functions $x \mapsto x^k \left( 1- x \right)^{m-k}$ are decreasing. 
Consequently 
\[
\p \left( \Bin \left( m, x_n^1 \right) = m \right) < \p \left( \Bin \left( m, x_n^2 \right) = m \right)
\]
and for $0 \leq k < m$, 
\[
\p \left( \Bin \left( m, x_n^1 \right) = k \right) > \p \left( \Bin \left( m, x_n^2 \right) = k \right), 
\]
where $\Bin \left(m, x \right)$ denotes a binomial random variable with parameters $m$ and $x$. 
Second, $P\left( 1 \right) = 0$ implies that $p_m = 1$. 
Repeated application of this coupling shows that for any $1 - \eps < y^1 < y^2 < 1$ we have 
\[
\p \left( A_n \, \middle| \, x_n = y^1 \right) \leq \p \left( A_n \, \middle| \, x_n = y^2 \right); 
\]
in particular, we have that 
\[
\p \left( A_{n_0} \, \middle| \, x_{n_0} = y \right) \geq \p \left( A_{n_0} \, \middle| \, x_{n_0} = y_0 \right) 
\]
for all $y \geq y_0$.

Now consider the auxiliary process. 
For one, we have 
$\p \left( \widetilde{x}_{n_0} \geq y_0 \right) > 0$. 
Moreover, if $x_{n_0} = \widetilde{x}_{n_0}$, on the event $A_{n_0}$ we can couple the processes 
$\left\{ x_n \right\}_{n \geq n_0}$ and $\left\{ \widetilde{x}_n \right\}_{n \geq n_0}$ 
so that $x_n = \widetilde{x}_n$ for all $n \geq n_0$, 
which shows that 
\[
\p \left( \widetilde{A}_{n_0} \, \middle| \, \widetilde{x}_{n_0} = y \right) \geq \p \left( A_{n_0} \, \middle| \, x_{n_0} = y_0 \right) > 0 
\]
for all $y \geq y_0$.
In particular, this shows that 
$\p \left( \lim_{n \to \infty} x_n = 1 \right) > 0$ implies that 
$\p \left( \lim_{n \to \infty} \widetilde{x}_n = 1 \right) > 0$. 
Thus it suffices to show that $\p \left( \lim_{n \to \infty} \widetilde{x}_n = 1 \right) = 0$. 

We claim that $\widetilde{x}_n \leq_{\icx} \overline{x}_n$ implies that $\p \left( \lim_{n \to \infty} \widetilde{x}_n = 1 \right) = 0$. 
To see this, for $\delta > 0$ define the function 
$g_\delta : \left[ 0, 1 \right] \to \left[ 0, 2 \right]$ by 
$g_\delta \left( x \right) = \max \left\{ 0, 2 - 1/\delta + x / \delta \right\}$. 
This is an increasing continuous convex function, 
and so $\widetilde{x}_n \leq_{\icx} \overline{x}_n$ implies that
\begin{equation}\label{eq:icx-g}
 \p \left( \widetilde{x}_n > 1 - \delta \right) \leq \E \left( g_\delta \left( \widetilde{x}_n \right) \right) \leq \E \left( g_\delta \left( \overline{x}_n \right) \right) \leq 2 \p \left( \overline{x}_n > 1 - 2 \delta \right).
\end{equation}
We know that the limiting distribution of $\overline{x}_n$ is a beta distribution, and thus
\[
 \lim_{\delta \to 0} \lim_{n \to \infty} \p \left( \overline{x}_n > 1 - 2 \delta \right) = 0.
\]
By~\eqref{eq:icx-g} this then implies that $\p \left( \lim_{n \to \infty} \widetilde{x}_n = 1 \right) = 0$.

We prove $\widetilde{X}_n \leq_{\icx} \overline{X}_n$ 
(which is equivalent to $\widetilde{x}_n \leq_{\icx} \overline{x}_n$) 
by induction on $n$; 
for $n=0$ this is immediate since the initial conditions agree. 
Fix now a positive integer $n$, and consider a random variable $X$ which attains integer values in the interval $\left[X_0, S_0 + 2mn \right]$, and let $x = \frac{X}{S_0 + 2mn}$. 
Denote by $\overline{X}$ a random variable with distribution 
$\p \left( \overline{X} = X + 2m \, \middle| \, X \right) = x$ and 
$\p \left( \overline{X} = X \, \middle| \, X \right) = 1 - x$. 
Similarly, let $\widetilde{X}$ denote a random variable with distribution the same as that of $\widetilde{X}_{n+1}$ conditioned on $\widetilde{X}_n = X$. 
Following Pemantle~\cite{pemantle1990time}, the induction step follows from the following two claims: 
(i) $\widetilde{X} \leq_{\icx} \overline{X}$, and 
(ii) $X \leq_{\icx} Y$ implies that $\overline{X} \leq_{\icx} \overline{Y}$.

First, it is enough to show that for any fixed $r$, conditioned on $x = r$ we have 
$\widetilde{X} \leq_{\icx} \overline{X}$; 
one can then integrate out the conditioning to get (i). 
We show this by checking the conditions of Lemma~\ref{lem:icx_suff}. 
First, when $r \leq 1 - \eps$ we have 
$\E \left( \widetilde{X} \, \middle| \, x = r \right) = \E \left( \overline{X} \, \middle| \, x = r \right)$ 
by the definition of the auxiliary process. 
If $r > 1 - \eps$ then we have 
$\E \left( \overline{X} \, \middle| \, x = r \right) = r \left( S_0 + 2mn \right) + 2mr$, 
while 
$\E \left( \widetilde{X} \, \middle| \, x = r \right) = r \left( S_0 + 2mn \right) + 2m \left( r + P \left( r \right) \right)$. 
Since $r > 1 - \eps$, 
$P\left( r \right) < 0$, 
and thus 
$\E \left( \widetilde{X} \, \middle| \, x = r \right) < \E \left( \overline{X} \, \middle| \, x = r \right)$. 
This shows that 
$\E \left( \widetilde{X} \, \middle| \, x = r \right) \leq \E \left( \overline{X} \, \middle| \, x = r \right)$. 
The other condition of Lemma~\ref{lem:icx_suff} holds automatically due to the fact that conditioned on $X = \ell$, 
the distribution of $\overline{X}$ is supported on the two values $\left\{ \ell, \ell + 2m \right\}$, 
while the support of the  distribution of $\widetilde{X}$ is contained in the interval $\left[ \ell, \ell + 2m\right]$. 

In view of Lemmas~\ref{lem:st-cx-icx} and~\ref{lem:Polya-cx}, 
to show (ii) it is enough to show that 
$X \leq_{\st} Y$ implies $\overline{X} \leq_{\st} \overline{Y}$, 
i.e., that for any increasing function $\phi$ we have 
$\E \left( \phi \left( \overline{X} \right) \right) \leq \E \left( \phi \left( \overline{Y} \right) \right)$. 
By conditioning on $X$ and $Y$, we have 
$\E \left( \phi \left( \overline{X} \right) \right) = \E \left( \overline{\phi} \left( X \right) \right)$ and 
$\E \left( \phi \left( \overline{Y} \right) \right) = \E \left( \overline{\phi} \left( Y \right) \right)$, 
where $\overline{\phi} \left( t \right) := \phi \left( t \right) \left( 1 - \alpha t \right) + \phi \left( t + 2m \right) \alpha t$, 
where $\alpha = \left( S_0 + 2mn \right)^{-1}$ and $t$ is such that $0 \leq \alpha t \leq 1$. 
Since $X \leq_{\st} Y$, we only need to show that $\overline{\phi}$ is increasing. 
Indeed, if $t_1 < t_2$ then
\begin{align*}
 \overline{\phi} \left( t_2 \right) - \overline{\phi} \left( t_1 \right) &= \left( \phi \left( t_2 + 2m \right) - \phi \left( t_1 + 2m \right) \right) \alpha t_1 + \left( \phi \left( t_2 \right) - \phi \left( t_1 \right) \right) \left( 1 - \alpha t_1 \right) \\
 &\quad + \left( \phi \left( t_2 + 2m \right) - \phi \left( t_2 \right) \right) \alpha \left( t_2 - t_1 \right),
\end{align*}
which is nonnegative, since all of the terms on the right hand side are nonnegative.
\end{proof}

\begin{proof}[Proof of Theorem~\ref{thm:nonlin_touchpoints}]
 This follows directly from Lemma~\ref{lem:stoch_appx} and Theorem~\ref{thm:conv_touchpoints}.
\end{proof}

%%%%%%%%%%%%%%%%%%%%%%%%%%%%%%%%%%%%%%%%%%%%%%%
\section{Many colours/types}\label{sec:many} %%%
%%%%%%%%%%%%%%%%%%%%%%%%%%%%%%%%%%%%%%%%%%%%%%%

It is both natural and important to study competition between more than two colours/types. 
The model naturally extends in this direction, and in this section we present our results regarding $N \geq 3$ competing types. 
In the following, vectors will be denoted using boldface, subscripts typically correspond to time and superscripts correspond to the indices of types. Furthermore, denote by $\Delta^N$ the probability simplex in $\R^N$.

The natural extension of the model to multiple competing types is as follows. At time zero, there is a graph $G_0$, where each node is of exactly one of the $N$ types. At each timestep a new node is added to the graph, and is connected to $m$ nodes of the original graph according to linear preferential attachment. The types of these $m$ neighbours induce a vector of types $\bfu$, where $u^i$ is the number of neighbours of type $i$. The type of the new node is then determined according to a random draw from the distribution $p_{\bfu} = \left\{ p_{\bfu}^i \right\}_{i \in \left[ N \right]}$. The probabilities $\left\{ p_{\bfu}^i \right\}_{\bfu,i}$ are parameters of the model.

As in the case of two types, our primary interest is in the fraction of nodes of each type. 
Let $A_n^i$ denote the number of nodes of type $i$ at time $n$, and let $\bfA_n = \left( A_n^1, \dots, A_n^N \right)$ denote the resulting vector of types. Let $\bfa_n$ denote the normalised vector of types, such that $\sum_{i=1}^N a_n^i = 1$. 
Furthermore, let $X_n^i$ denote the sum of the degrees of type $i$ nodes at time $n$, let $\bfX_n = \left( X_n^1, \dots, X_n^N \right)$ denote the resulting vector of degrees, and let $\bfx_n$ be the normalised vector of degrees, such that $\sum_{i=1}^N x_n^i = 1$.

As in the $N=2$ case, there is a clear distinction between the \emph{linear model}, when $p_{\bfu}^i = \frac{u^i}{m}$ for all $\bfu$ and $i \in \left[ N \right]$, and \emph{nonlinear models}, when there exist $\bfu$ and $i \in \left[ N \right]$ such that $p_{\bfu}^i \neq \frac{u^i}{m}$. In fact, the linear model for $N \geq 3$ types reduces to the linear model for two types. This is because in the linear model, if we want to study the evolution of the size of type $i$, then we can group all other types into a single ``mega-type'', denoted by $-i$, and run the process with two types: type $i$ and ``mega-type'' $-i$. Due to linearity, the original process with $N$ types and the process with type $i$ and ``mega-type'' $-i$ can be coupled such that the evolution of type $i$ is identical in the two processes. Consequently, in the linear model all the results of the $N = 2$ case apply. In particular, we have the following theorem.

\begin{theorem}[Linear model]\label{thm:many_lin}
Assume that $p_{\bfu}^i = \frac{u^i}{m}$ for all $\bfu$ and $i \in \left[ N \right]$, and that $X_0^i > 0$ for all $i \in \left[ N \right]$. Then $\bfa_n$ converges almost surely, and the limiting distribution has full support on $\Delta^N$, and no atoms.\hfill$\openbox$
\end{theorem}

In nonlinear models, as we will see later, a key role in the asymptotic behaviour of the process $\left\{ \bfa_n \right\}_{n \geq 0}$ is played by the vector field
\begin{equation}\label{eq:vec_field}
 \bfP \left( \bfy \right) = \frac{1}{2} \sum_{i=1}^N \sum_{\bfu} \binom{m}{\bfu} \left( \bfy \right)^{\bfu} \left[ p_{\bfu}^i - \frac{u^i}{m} \right] \bfmath{\delta}^i,
\end{equation}
where $\binom{m}{\bfu} = \frac{m!}{u^1! \dots u^N!}$ denotes the multinomial coefficient, $\left( \bfy \right)^{\bfu} = \left( y^1 \right)^{u^1} \left( y^2 \right)^{u^2} \dots \left( y^N \right)^{u^N}$, and $\bfmath{\delta}^i$ is the $N$-dimensional unit vector whose $i^{\text{th}}$ coordinate is $1$ and all other coordinates are $0$. Let us denote the zero set of this vector field on the probability simplex by $Z_{\bfP} := \left\{ \bfy \in \Delta^N : \bfP \left( \bfy \right) = \bfmath{0} \right\}$; this will be important later.

The behaviour of the process in the general nonlinear model with multiple types is involved, and its complete theoretical analysis is as of yet out of our reach. Nonetheless, based on partial theoretical results, we conjecture the following asymptotic behaviour, which is similar to that in the case of two types.

\begin{conjecture}[Nonlinear models]\label{conj:many_nonlin}
 Assume that there exist $\bfu$ and $i \in \left[ N \right]$ such that $p_{\bfu}^i \neq \frac{u^i}{m}$, and that $X_0^i > 0$ for all $i \in \left[ N \right]$. Then $\bfa_n$ converges almost surely and the limit is a point in the zero set $Z_{\bfP}$.
\end{conjecture}

In the rest of this section we describe theoretical progress towards this conjecture. As in the case of two competing types, the problem can be cast in a (multidimensional) stochastic approximation framework.

The process $\left\{ \bfA_n \right\}_{n \geq 0}$ is not a Markov process, and therefore we study the joint process $\left\{ \left( \bfA_n, \bfX_n \right) \right\}_{n \geq 0}$, which is indeed Markov. It evolves as follows. Given $\left( \bfA_n, \bfX_n \right)$, a vector $\bfu_{n+1}$ is drawn from the multinomial distribution with parameters $m$ and $\bfx_n$. Subsequently, an index $I_{n+1} \in \left[N \right]$ is chosen from the distribution $p_{\bfu_{n+1}}$. We then have
\begin{align}
 \bfA_{n+1} &= \bfA_n + \bfmath{\delta}^{I_{n+1}}, \label{eq:An_evol}\\
 \bfX_{n+1} &= \bfX_n + \bfu_{n+1} + m \bfmath{\delta}^{I_{n+1}}. \label{eq:Xn_evol}
\end{align}

Before analysing the process $\left\{ \bfx_n \right\}_{n \geq 0}$, let us show that in order to prove Conjecture~\ref{conj:many_nonlin} on the asymptotic behaviour of $\left\{ \bfa_n \right\}_{n \geq 0}$, it is sufficient to prove a similar result on the asymptotic behaviour of $\left\{ \bfx_n \right\}_{n \geq 0}$.

\begin{lemma}\label{lem:many_degree}
Assume that there exist $\bfu$ and $i \in \left[ N \right]$ such that $p_{\bfu}^i \neq \frac{u^i}{m}$, and that $X_0^i > 0$ for all $i \in \left[ N \right]$. 
Assume that $\bfx_n$ converges almost surely and the limit is a point in the zero set $Z_{\bfP}$. 
Then $\bfa_n$ converges almost surely and 
$\lim_{n \to \infty} \bfa_n = \lim_{n \to \infty} \bfx_n \in Z_{\bfP}$. 
\end{lemma}

\begin{proof}
This is similar to the proof of Lemma~\ref{lem:from_half_edges_to_vertices}. We have seen that
\begin{align*}
 \E \left( \bfA_{n+1} - \bfA_n \, \middle| \, \mcF_n \right) &= \E \left( \bfmath{\delta}^{I_{n+1}} \, \middle| \, \mcF_n \right) = \sum_{i=1}^N \sum_{\bfu} \binom{m}{\bfu} \left( \bfx_n \right)^{\bfu} p_{\bfu}^i \bfmath{\delta}^i\\
&= \bfx_n + \sum_{i=1}^N \sum_{\bfu} \binom{m}{\bfu} \left( \bfx_n \right)^{\bfu} \left[ p_{\bfu}^i - \frac{u^i}{m} \right] \bfmath{\delta}^i = \bfx_n + 2 \bfP \left( \bfx_n \right) =: f \left( \bfx_n \right).
\end{align*}
Let $\bfM_0 = \bfmath{0}$ and define the martingale $\bfM_n := \bfA_n - \bfA_0 - \sum_{j=0}^{n-1} f \left( \bfx_j \right)$. 
This martingale has bounded increments, and thus $\lim_{n \to \infty} \bfM_n / n = \bfmath{0}$ a.s. 
By the definition of the martingale, this shows that a.s.
\[
 \lim_{n \to \infty} \left[ \bfa_n - \frac{1}{n} \sum_{j=0}^{n-1} f \left( \bfx_j \right) \right] = \bfmath{0}.
\]
Now if the limit $\lim_{n \to \infty} \bfx_n$ exists a.s., and any limit point $\bfx$ satisfies $\bfP \left( \bfx \right) = \bfmath{0}$, then also $f\left( \bfx \right) = \bfx$, and thus the limit of the Ces\`aro mean of the sequence $\left\{ f \left( \bfx_n \right) \right\}_{n \geq 0}$ also converges to the same limit point. This then implies that the limit $\lim_{n \to \infty} \bfa_n$ exists a.s.\ and is equal to $\lim_{n \to \infty} \bfx_n$.
\end{proof}

The key observation in the analysis of the asymptotic behaviour of $\left\{ \bfx_n \right\}_{n \geq 0}$ is that it is a stochastic approximation process. 
In higher dimensions, a stochastic approximation process is defined as follows. Let $\bfZ_n$ be a stochastic process in the Euclidean space $\R^N$ and adapted to a filtration $\left\{ \mcF_n \right\}_{n\geq 0}$. Suppose that it satisfies
\[
 \bfZ_{n+1} - \bfZ_n = \frac{1}{n} \left( \bfF \left( \bfZ_n \right) + \bfmath{\xi}_{n+1} + \bfR_n \right),
\]
where $\bfF$ is a vector field on $\R^N$, $\E \left( \bfmath{\xi}_{n+1} \, \middle| \, \mcF_n \right) = \bfmath{0}$, and the remainder terms $\bfR_n \in \mcF_n$ go to zero and satisfy $\sum_{n=1}^{\infty} n^{-1} \left\| \bfR_n \right\| < \infty$ a.s. Such a process is known as a stochastic approximation process.

\begin{lemma}\label{lem:many_stoch_appx}
 The process $\left\{ \bfx_n \right\}_{n \geq 0}$ is a stochastic approximation process with the vector field $\bfP$ as in~\eqref{eq:vec_field}. Furthermore, the noise term $\bfmath{\xi}_{n}$ is bounded: $\left\| \bfmath{\xi}_n \right\|_1 \leq 2N$ for all $n \geq 1$.
\end{lemma}

\begin{proof}
 From \eqref{eq:Xn_evol} we have that $\E \left( \bfX_{n+1} - \bfX_n \, \middle| \, \mcF_n \right) = \E \left( \bfu_{n+1} \, \middle| \, \mcF_n \right) + m \E \left( \bfmath{\delta}^{I_{n+1}} \, \middle| \, \mcF_n \right)$. Given $\mcF_n$, $\bfu_{n+1}$ is multinomial with parameters $m$ and $\bfx_n$, and so $\E \left( \bfu_{n+1} \, \middle| \, \mcF_n \right) = m \bfx_n$. By construction, we have that
\[
 \E \left( \bfmath{\delta}^{I_{n+1}} \, \middle| \, \mcF_n \right) = \sum_{i=1}^N \sum_{\bfu} \binom{m}{\bfu} \left( \bfx_n \right)^{\bfu} p_{\bfu}^i \bfmath{\delta}^i.
\]
 Let $S_0$ denote the sum of the degrees in $G_0$, and let $S_n = S_0 + 2mn$. A simple calculation gives that $\bfx_{n+1} - \bfx_n = \frac{\bfX_{n+1} - \bfX_n - 2m \bfx_n}{S_{n+1}}$, and so we have
\[
 \E \left( \bfx_{n+1} - \bfx_n \, \middle| \, \mcF_n \right) = \frac{m}{S_{n+1}} \sum_{i=1}^N \sum_{\bfu} \binom{m}{\bfu} \left( \bfx_n \right)^{\bfu} \left[ p_{\bfu}^i - \frac{u^i}{m} \right] \bfmath{\delta}^i = \frac{2m}{S_{n+1}} \bfP \left( \bfx_n \right).
\]
We can then write $\left\{ \bfx_n \right\}_{n \geq 0}$ as a stochastic approximation process:
\[
 \bfx_{n+1} - \bfx_n = \frac{1}{n} \left[ \bfP \left( \bfx_n \right) + \bfmath{\xi}_{n+1} + \bfmath{R}_n \right],
\]
where
\begin{equation}\label{eq:noise_term}
 \bfmath{\xi}_{n+1} = n \left\{ \bfx_{n+1} - \bfx_n - \E \left( \bfx_{n+1} - \bfx_n \, \middle| \, \mcF_n \right) \right\}
\end{equation}
is the martingale term, and the remainder term is
\[
 \bfmath{R}_n = - \frac{S_0 + 2m}{S_0 + 2m \left( n + 1 \right)} \bfP \left( \bfx_n \right).
\]
Clearly $\bfR_n \in \mcF_n$ and similarly as at the end of the proof of Lemma~\ref{lem:stoch_appx} one can show that $\left\| \bfR_n \right\| \leq c / n$ for some constant $c = c\left( N, S_0, m \right)$, which implies that $\sum_{n=1}^{\infty} n^{-1} \left\| \bfR_n \right\| < \infty$ a.s.

Finally, to check that $\left\| \bfmath{\xi}_n \right\|_1 \leq 2N$, note that
\[
 \left| x_{n+1}^i - x_n^i \right| = \left| \frac{X_{n+1}^i - X_n^i - 2m x_n^i}{S_{n+1}} \right| \leq \frac{2m}{2m\left( n + 1 \right)} = \frac{1}{n+1},
\]
and then use~\eqref{eq:noise_term}.
\end{proof}

As in the one-dimensional case, intuitively, trajectories of a stochastic approximation process $\left\{ \bfZ_n \right\}_{n \geq 0}$ should approximate the trajectories $\left\{ \bfZ \left( t \right) \right\}_{t \geq 0}$ of the corresponding ODE $d\bfZ / dt = \bfF \left( Z \right)$. Moreover, since $\left\{ \bfZ_n \right\}_{n \geq 0}$ is a stochastic system, we expect that stable trajectories of the ODE should appear, but unstable trajectories should not. 

The main concept in formalizing this intuition is that of an \emph{asymptotic pseudotrajectory}, introduced by Bena{\"\i}m and Hirsch~\cite{benaim1996asymptotic}. We omit the precise definition, and refer to Bena\"im's lecture notes on the topic for more details~\cite{benaim1999dynamics} (see also~\cite[Section 2.5]{pemantle2007survey} for a concise summary). There are many results that give sufficient conditions for a stochastic approximation process to be an asymptotic pseudotrajectory of the corresponding ODE. In particular,~\cite[Proposition 4.4 and Remark 4.5]{benaim1999dynamics} (see also~\cite[Theorem 2.13]{pemantle2007survey}), together with Lemma~\ref{lem:many_stoch_appx} and the fact that $\bfP$ is Lipschitz, imply the following.
\begin{corollary}\label{cor:stoch_appx_asymptotic_psuedotrajectory}
Let $\left\{ \bfx \left( t \right) \right\}_{t \geq 0}$ linearly interpolate $\left\{ \bfx_n \right\}_{n \geq 0}$ at nonintegral times. 
Then $\left\{ \bfx \left( t \right) \right\}_{t \geq 0}$ is a.s.\ an asymptotic pseudotrajectory for the flow induced by the vector field $\bfP$ via the ODE $d\bfy / dt = \bfP \left( \bfy \right)$.\hfill$\openbox$
\end{corollary}
There are further general results about asymptotic pseudotrajectories that apply to the stochastic approximation process $\left\{ \bfx_n \right\}_{n \geq 0}$, e.g., about convergence to attractors and nonconvergence to linearly unstable equilibria. 
However, we omit these, as we prefer to emphasise the main message of Corollary~\ref{cor:stoch_appx_asymptotic_psuedotrajectory}. The main point is that in order to understand the stochastic approximation process $\left\{ \bfx_n \right\}_{n \geq 0}$, we need to understand the vector field $\bfP$, and the corresponding ODE
\[
 \frac{d\bfy}{dt} = \bfP \left( \bfy \right).
\]

Unfortunately, understanding the behaviour of such nonlinear ODEs is a notoriously difficult subject (see, e.g., the book by Hirsch, Smale and Devaney~\cite{hirsch2004differential}). The most successful tool in this area is Lyapunov theory (see, e.g., the recent paper~\cite{benaim2013generalized}), and this can indeed be applied to our problem for special values of the parameters; however, it seems difficult to apply this theory to the vector field $\bfP$ for \emph{generic} values of the parameters $\left\{ p_{\bfu}^i \right\}_{\bfu,i}$.

For instance, if $\bfP$ is a gradient, i.e., $\bfP = - \nabla V$ for some $V : \R^N \to \R$, 
then Corollary~\ref{cor:stoch_appx_asymptotic_psuedotrajectory} and general results about asymptotic pseudotrajectories (see~\cite{benaim1999dynamics}) imply that 
$\bfx_n$ converges almost surely and the limit is a point in the zero set $Z_{\bfP}$, 
which, by Lemma~\ref{lem:many_degree}, implies that Conjecture~\ref{conj:many_nonlin} holds. 
An example of when $\bfP$ is a gradient is when the probability of the new node adopting type $i$ depends only on the number of type $i$ connections, 
i.e., $p_{\bfu}^i = \phi \left( u^i \right)$ for some function $\phi$ which does not depend on $i$. 
This implies that $\phi$ must be of the form $\phi \left( z \right) = \alpha \frac{1}{N} + \left( 1 - \alpha \right) \frac{z}{m}$ for some $0 \leq \alpha \leq 1$,\footnote{To see this, first observe that there are $m+1$ free parameters when constructing such a function: 
$\phi \left( 0 \right)$, $\phi \left( 1 \right)$, $\dots$, $\phi \left( m \right)$. 
These parameters also satisfy constraints, since the probabilities have to sum to $1$: $\sum_{i=1}^N p_{\bfu}^i = 1$ for every vector of types $\bfu$. 
In particular, the following $m$ equations are linearly independent: 
$\left( N - 3 \right) \phi \left( 0 \right) + \phi \left( 0 \right) + \phi \left( i \right) + \phi \left( m - i \right) = 1$, 
for $i = 0, 1, \dots, \left\lfloor m/2 \right\rfloor$, 
and  
$\left( N - 3 \right) \phi \left( 0 \right) + \phi \left( 1 \right) + \phi \left( i \right) + \phi \left( m - i - 1 \right) = 1$, 
for $i = 1, 2, \dots, \left\lfloor \left( m - 1 \right) / 2 \right\rfloor$. 
Consequently, there can be at most $2$ linearly independent solutions to these equations. 
The constant function $\phi \left( z \right) = 1/N$ and the linear function $\phi \left( z \right) = z / m$ are solutions, 
which means that any solution is a linear combination of these. 
The restriction of $\alpha$ to be in $\left[ 0, 1 \right]$ follows from the constraint that the probabilities be nonnegative.
} 
which corresponds to a mixture of the linear model and a uniformly random choice. 
In this case $\bfP \left( \bfy \right) = \frac{\alpha}{2} \left( \frac{1}{N} \mathbf{1} - \bfy \right)$, where $\mathbf{1} \in \R^N$ is the vector with all entries equal to $1$, and thus when $\alpha > 0$ then $\bfa_n$ converges a.s.\ to $\frac{1}{N} \mathbf{1}$.

However, for \emph{generic} parameter values $\bfP$ will not be a gradient. To see this, note that $\bfP$ being a gradient implies that
\begin{equation}\label{eq:grad}
 \pd{\left( \bfP \left( \bfy \right) \right)^i}{y^j} = \pd{\left( \bfP \left( \bfy \right) \right)^j}{y^i}
\end{equation}
for every $i \neq j$. Without any restrictions, there are $\left( N - 1 \right) \binom{m+N-1}{N-1}$ free parameters in $\bfP$. The gradient condition~\eqref{eq:grad} imposes an additional $\binom{N}{2}$ constraints, which will not be satisfied for generic parameter values.

%%%%%%%%%%%%%%%%%%%%%%%%%%%%%%%%%%%%%%%%%%%%%%%%%%%%%%%%%%%%%%%%%%%
\section{Open problems and future directions}\label{sec:future} %%%
%%%%%%%%%%%%%%%%%%%%%%%%%%%%%%%%%%%%%%%%%%%%%%%%%%%%%%%%%%%%%%%%%%%

Our paper leaves open several interesting problems. Two immediate open questions concerning the model are the following.
\begin{itemize}
 \item 
\textbf{Limiting distribution in the linear model for two types.} Our Theorem~\ref{thm:continuous_case} gives us information about the limiting behaviour of $\left\{ a_n \right\}_{n \geq 0}$, but it does not identify the distribution of $a := \lim_{n \to \infty} a_n$.

 For $m=1$ the process $\left\{x_n \right\}_{n \geq 0}$ corresponds to a P\'olya urn where whenever one draws a ball, one puts back two extra balls of the same colour. This is because when a new node joins the graph, its colour automatically becomes the colour of its initial neighbour. Thus the distribution of $x$---and by Lemma~\ref{lem:from_half_edges_to_vertices} the distribution of $a$ as well---is the Beta distribution with parameters $\frac{X_0}{2}$ and $\frac{Y_0}{2}$.

 However, for $m > 1$ we do not know what the limiting distribution is. Note that simulations show that the limiting distribution can be bimodal; see, e.g., Figure~\ref{fig:hist2}.

 \item
 \textbf{Understanding the vector field $\bfP$.} As discussed in Section~\ref{sec:many}, in order to understand the behaviour of the general nonlinear model in the case of multiple types---and in order to prove or disprove Conjecture~\ref{conj:many_nonlin}---a good understanding of the vector field $\bfP$ and the corresponding ODE $d\bfy / dt = \bfP \left( \bfy \right)$ is needed. We leave this as our second open problem.
\end{itemize}

A key property of the model is its \emph{simplicity}. 
However, this also means that certain aspects of real-world networks and processes influencing product adoption are simplified or not considered. 
It would be interesting to understand the following possible extensions of the model, and, in particular, whether anything can be said analytically in these extensions. 

\begin{itemize}
 \item 
 \textbf{Changing preferences.} In the model once a node receives a type, that type is then fixed and cannot change over time. A possible extension of the model is to allow the type of a node to change over time. This can model changing preferences of individuals, e.g., somebody moving from one mobile phone provider to another.

 \item
 \textbf{Allowing multiple types for a single individual.} In the model a node can only have a single type. This is reasonable in many situations (e.g., an individual typically has only one mobile phone provider), but modeling other situations might require allowing nodes to simultaneously have multiple types.

 \item 
 \textbf{Other network evolution models.} The preferential attachment model is a reasonable approximation of some real-world networks, and it has the advantageous property of being analytically tractable.  How does the model behave under other network evolution models? Can similar results be shown analytically/experimentally? Are the results robust to small changes in the network evolution model?

In particular, the affine preferential attachment model is discussed at the end of Section~\ref{sec:results}. 
This model is more involved, because when $c \neq 0$, 
then $\left\{ X_n \right\}_{n \geq 0}$ is not a Markov process, 
so one has to understand the joint evolution of 
$\left\{ \left( A_n, X_n \right) \right\}_{n \geq 0}$. 
A calculation shows that 
$\left\{ \left( a_n, x_n \right) \right\}_{n \geq 0}$ 
is a stochastic approximation process with vector field
\[
 \bfF \left( a, x \right) = 
\left(
 \tfrac{2m}{2m+c} \left( x - a \right) + 2 P \left( \tfrac{2m}{2m+c} x + \tfrac{c}{2m+c} a \right),
 \tfrac{c}{2m+c} \left( a - x \right) +  P \left( \tfrac{2m}{2m+c} x + \tfrac{c}{2m+c} a \right) 
\right).
\]
For $c \geq 0$, the zero set of $\bfF$ is exactly 
$\left\{ \left( a, x \right) : a = x, P \left( x \right) = 0 \right\}$, 
and so the same results as in Section~\ref{sec:results} apply. 
For $c \in \left( -1, 0 \right)$ the zero set of $\bfF$ could be larger; 
we leave the analysis of this as an exercise to the reader. 

The case of uniform attachment is even simpler than preferential attachment, 
since then $\left\{ A_n \right\}_{n \geq 0}$ is a Markov process, 
and it is unnecessary to consider the sum-of-degrees process $\left\{ X_n \right\}_{n \geq 0}$. 
The same results as in Section~\ref{sec:results} apply. 

We leave the study of further network evolution models for future work.

 \item
 \textbf{Other type adoption mechanisms.} The model incorporates a fairly general type adoption mechanism, but various modifications would be interesting to explore. For instance, in real life choices are often made based on the opinions of specific friends, not just based on aggregate information of one's friends. 

 \item
 \textbf{Marketing.} In essence, the model describes word-of-mouth recommendations, and does not consider marketing efforts by the competing companies, such as advertising. How does incorporating marketing affect the results?

\end{itemize}

In conclusion, through a simple model we have coupled network evolution and type adoption, leading to a possible explanation of coexistence in preferential attachment networks. Exploring various modifications and extensions of this model, such as those mentioned above, will be crucial in determining the robustness of this phenomenon, and will help elucidate our understanding of these processes.

%%%%%%%%%%%%%%%%%%%%%%%%
%%% Acknowledgements %%%
%%%%%%%%%%%%%%%%%%%%%%%%

\section*{Acknowledgments}

We thank Erik Bodzs{\'a}r, Gy{\"o}rgy Korniss, G{\'e}za Mesz{\'e}na, Jasmine Nirody, Nathan Ross, Allan Sly, and Isabelle Stanton for helpful discussions. 
We also thank Oliver Riordan and an anonymous referee for useful comments on the manuscript.

%%%%%%%%%%%%%%%%%%
%%% References %%%
%%%%%%%%%%%%%%%%%%

\bibliographystyle{plain}
\bibliography{coex}

\begin{thebibliography}{10}

\bibitem{antunovic2011competing}
T.~Antunovi{\'c}, Y.~Dekel, E.~Mossel, and Y.~Peres.
\newblock {Competing first passage percolation on random regular graphs}.
\newblock {\em Arxiv preprint arXiv:1109.2575}, 2011.

\bibitem{arthur1990positive}
W.B. Arthur.
\newblock {Positive Feedbacks in the Economy}.
\newblock {\em Sci. Am.}, 262(2):92--99, 1990.

\bibitem{arthur1994increasing}
W.B. Arthur.
\newblock {\em {Increasing Returns and Path Dependence in the Economy}}.
\newblock The University of Michigan Press, 1994.

\bibitem{banerjee2004word}
A.~Banerjee and D.~Fudenberg.
\newblock Word-of-mouth learning.
\newblock {\em Game. Econ. Behav.}, 46(1):1--22, 2004.

\bibitem{barabasi1999emergence}
A.L. Barab{\'a}si and R.~Albert.
\newblock Emergence of scaling in random networks.
\newblock {\em Science}, 286(5439):509--512, 1999.

\bibitem{barrat2008dynamical}
A.~Barrat, M.~Barth{\'e}lemy, and A.~Vespignani.
\newblock {\em {Dynamical Processes on Complex Networks}}.
\newblock {Cambridge University Press}, 2008.

\bibitem{benaim1999dynamics}
M.~Bena{\"\i}m.
\newblock Dynamics of stochastic approximation algorithms.
\newblock {\em S{\'e}minaire de probabilit{\'e}s XXXIII}, pages 1--68, 1999.

\bibitem{benaim2013generalized}
M.~Bena{\"\i}m, I.~Benjamini, J.~Chen, and Y.~Lima.
\newblock {A generalized P{\'o}lya's urn with graph based interactions}.
\newblock {\em Random Struct. Alg.}, 46(4):614--634, 2015.

\bibitem{benaim1996asymptotic}
M.~Bena{\"\i}m and M.W. Hirsch.
\newblock {Asymptotic Pseudotrajectories and Chain Recurrent Flows, with
  Applications}.
\newblock {\em J. Dynam. Differential Equations}, 8(1):141--176, 1996.

\bibitem{BeBoChSa:14}
N.~Berger, C.~Borgs, J.T. Chayes, and A.~Saberi.
\newblock Asymptotic behavior and distributional limits of preferential
  attachment graphs.
\newblock {\em The Annals of Probability}, 42(1):1--40, 2014.

\bibitem{bollobas2004diameter}
B.~Bollob{\'a}s and O.~Riordan.
\newblock The diameter of a scale-free random graph.
\newblock {\em Combinatorica}, 24(1):5--34, 2004.

\bibitem{deijfen2013winner}
M.~Deijfen and R.~van~der Hofstad.
\newblock {The winner takes it all}.
\newblock {\em Arxiv preprint arXiv:1306.6467}, 2013.

\bibitem{gross2008adaptive}
T.~Gross and B.~Blasius.
\newblock Adaptive coevolutionary networks: a review.
\newblock {\em J. R. Soc. Interface}, 5(20):259--271, 2008.

\bibitem{hill1980strong}
B.M. Hill, D.~Lane, and W.~Sudderth.
\newblock A strong law for some generalized urn processes.
\newblock {\em Ann. Probab.}, 8(2):214--226, 1980.

\bibitem{hirsch2004differential}
M.W. Hirsch, S.~Smale, and R.L. Devaney.
\newblock {\em {Differential Equations, Dynamical Systems, and An Introduction
  to Chaos}}.
\newblock {Academic Press}, 2004.

\bibitem{holme2012temporal}
P.~Holme and J.~Saram{\"a}ki.
\newblock Temporal networks.
\newblock {\em Phys. Rep.}, 519(3):97--125, 2012.

\bibitem{lelarge2012diffusion}
M.~Lelarge.
\newblock Diffusion and cascading behavior in random networks.
\newblock {\em Games and Economic Behavior}, 75(2):752--775, 2012.

\bibitem{nevelson1976stochastic}
M.B. Nevelson and R.Z. Hasminskii.
\newblock {\em {Stochastic Approximation and Recursive Estimation}}, volume~47
  of {\em {Translations of Mathematical Monographs}}.
\newblock {American Mathematical Society}, 1976.

\bibitem{ohtsuki2006simple}
H.~Ohtsuki, C.~Hauert, E.~Lieberman, and M.A. Nowak.
\newblock A simple rule for the evolution of cooperation on graphs and social
  networks.
\newblock {\em Nature}, 441(7092):502--505, 2006.

\bibitem{pastor2001epidemic}
R.~Pastor-Satorras and A.~Vespignani.
\newblock Epidemic spreading in scale-free networks.
\newblock {\em {Phys. Rev. Lett.}}, 86(14):3200--3203, 2001.

\bibitem{pemantle1990time}
R.~Pemantle.
\newblock {A time-dependent version of P{\'o}lya's urn}.
\newblock {\em J. Theoret. Probab.}, 3(4):627--637, 1990.

\bibitem{pemantle1990nonconvergence}
R.~Pemantle.
\newblock Nonconvergence to unstable points in urn models and stochastic
  approximations.
\newblock {\em Ann. Probab.}, 18(2):698--712, 1990.

\bibitem{pemantle1991touchpoints}
R.~Pemantle.
\newblock {When are touchpoints limits for generalized P{\'o}lya urns?}
\newblock {\em Proc. Amer. Math. Soc.}, 113(1):235--243, 1991.

\bibitem{pemantle2007survey}
R.~Pemantle.
\newblock A survey of random processes with reinforcement.
\newblock {\em Probab. Surv.}, 4:1--79, 2007.

\bibitem{prakash2012winner}
B.A. Prakash, A.~Beutel, R.~Rosenfeld, and C.~Faloutsos.
\newblock Winner takes all: competing viruses or ideas on fair-play networks.
\newblock In {\em {Proc. 21st Int. Conf. World Wide Web (WWW)}}, pages
  1037--1046. ACM, 2012.

\bibitem{redner1998popular}
S.~Redner.
\newblock How popular is your paper? an empirical study of the citation
  distribution.
\newblock {\em Eur. Phys. J. B}, 4(2):131--134, 1998.

\bibitem{robbins1951stochastic}
H.~Robbins and S.~Monro.
\newblock A stochastic approximation method.
\newblock {\em Ann. Math. Stat.}, 22(3):400--407, 1951.

\bibitem{shaked2007stochastic}
M.~Shaked and J.G. Shanthikumar.
\newblock {\em {Stochastic Orders}}.
\newblock Springer, {New York}, 2007.

\bibitem{skyrms2000dynamic}
B.~Skyrms and R.~Pemantle.
\newblock A dynamic model of social network formation.
\newblock {\em Proc. Nat. Acad. Sci. U.S.A.}, 97(16):9340--9346, 2000.

\bibitem{watts2002simple}
D.J. Watts.
\newblock A simple model of global cascades on random networks.
\newblock {\em Proc. Nat. Acad. Sci. U.S.A.}, 99(9):5766--5771, 2002.

\end{thebibliography}

%%%%%%%%%%%%%%%%
%%% Appendix %%%
%%%%%%%%%%%%%%%%

% \appendix

%%%%%%%%%%%%%%%%%%
\end{document}